\newif\ifdouble
\doubletrue
\ifdouble
\documentclass[journal]{IEEEtran}
\else
\documentclass[12pt, draftclsnofoot, onecolumn]{IEEEtran}
\fi
\newif \iftifs
\tifsfalse
\newif \ifNSM
\NSMtrue
\newif \ifCH
\CHfalse
\usepackage{amsthm}
\usepackage{latexsym}
\usepackage{amssymb}
\usepackage{graphicx}
\usepackage{caption}
\usepackage{subcaption}
\usepackage{amsmath}
\usepackage{color}
\usepackage{cite}
\usepackage{cleveref}
\usepackage[english]{babel}
\usepackage{footmisc}
\usepackage{enumerate}

\newenvironment{itemize*}%
  {\begin{itemize}%
    \setlength{\itemsep}{0pt}%
    \setlength{\parskip}{0pt}}%
  {\end{itemize}}

\theoremstyle{plain}
\newtheorem{theorem}{Theorem}
\newtheorem{lemma}{Lemma}

\newtheorem{claim}{Claim}

\theoremstyle{definition}
\newtheorem{definition}{Definition}
\theoremstyle{remark}
\newtheorem{remark}{Remark}

\newcommand{\off}[1]{}
\newcommand{\tifs}[1]{\textcolor[rgb]{0.00,0.00,0.00}{#1}}

\ifCLASSINFOpdf
\else
\fi

\begin{document}
\ifCH\else
\title{Secure Group Testing
\thanks{A. Cohen, A. Cohen and O. Gurewitz are with the Department of Communication Systems Engineering, Ben-Gurion University of the Negev, Beer-Sheva 84105, Israel (e-mail: alejandr@post.bgu.ac.il; coasaf@bgu.ac.il; gurewitz@bgu.ac.il).
\newline
Parts of this work were presented at the IEEE International Symposium on Information Theory, ISIT 2016.}}
\markboth{}{}
\ifdouble
\author{\IEEEauthorblockN{Alejandro Cohen\hspace{15 mm} Asaf Cohen \hspace{15 mm}  Omer Gurewitz} \vspace{-8mm}}
\else
\author{\IEEEauthorblockN{Alejandro Cohen \hspace{10 mm} Asaf Cohen \hspace{10 mm} Omer Gurewitz} \vspace{-8mm}}
\fi
\maketitle
\begin{abstract}
The principal goal of \emph{Group Testing} (GT) is to identify a small subset of ``defective" items from a large population, by grouping items into as few test pools as possible.
The test outcome of a pool is positive if it contains at least one defective item, and is negative otherwise.
GT algorithms are utilized in numerous applications, and in many of them maintaining the privacy of the tested items, namely, keeping secret whether they are defective or not, is critical.

In this paper, we consider a scenario where there is an eavesdropper (Eve) who is able to observe a subset of the GT outcomes (pools).
We propose a new non-adaptive \emph{Secure Group Testing} (SGT) \tifs{scheme} based on information-theoretic principles. \tifs{The new proposed test design} keeps the eavesdropper ignorant regarding the items' status. Specifically, when the fraction of tests observed by Eve is $0 \leq \delta <1$, we prove that \tifs{with the naive Maximum Likelihood (ML) decoding algorithm} the number of tests required for both correct reconstruction at the legitimate user (with high probability) and negligible information leakage to Eve is $\frac{1}{1-\delta}$ times the number of tests required with no secrecy constraint for the fixed $K$ regime.
By a matching converse, we completely characterize the Secure GT capacity.
Moreover, we consider \tifs{the Definitely Non-Defective (DND)} computationally efficient decoding algorithm, \tifs{proposed in the literature for non-secure GT. We prove that with the new secure test design}, for $\delta < 1/2$, the number of tests required, without any constraint on $K$, is at most $\frac{1}{1/2-\delta}$ times the number of  tests required with no secrecy constraint.
\end{abstract}
      %
\section{Introduction}\label{intro}
The classical version of Group Testing (GT) was suggested during World War II in order to identify syphilis-infected draftees while dramatically reducing the number of required tests \cite{dorfman1943detection}. Specifically, when the number of infected draftees, $K$, is much smaller than the population size, $N$, instead of examining each blood sample individually, one can conduct a small number of of \emph{pooled samples}. Each pool outcome is negative if it contains no infected sample, and positive if it contains at least one infected sample. The problem is thus to identify the infected draftees via as few pooled tests as possible. \Cref{basic_testing} (a)-(c) depicts a small example.

Since its exploitation in WWII, GT has been utilized in numerous fields, including biology and chemistry \cite{du1999combinatorial,macula1999probabilistic}, communications \cite{varanasi1995group,cheraghchi2012graph,wu2015partition,wu2014achievable}, sensor networks \cite{bajwa2007joint}, pattern matching \cite{clifford2007k} and web services \cite{tsai2004testing}. GT has also found applications in the emerging field of Cyber Security, e.g., detection of significant changes in network traffic \cite{cormode2005s}, Denial of Service attacks \cite{xuan2010detecting} and indexing information for data forensics \cite{goodrich2005indexing}. \tifs{Recently, GT was also considered for testing coronavirus (COVID-19) using significatly less tests than the number of tested subjects \cite{narayanan2020accelerated,theagarajan2020group,seong2020group}.}

Many scenarios which utilize GT involve sensitive information which should not be revealed if some of the tests leak (for instance, if one of the several labs to which tests have been distributed for parallel processing is compromised). However, in GT, leakage of even a single pool-test outcome may reveal significant information about the tested items. If the test outcome is negative it indicates that none of the items in the pool is defective; if it is positive, at least one of the items in the pool is defective (see \Cref{basic_testing} (d) for a short example). Accordingly, it is critical to ensure that a leakage of a fraction of the pool-tests outcomes to undesirable or malicious eavesdroppers does not give them any useful information on the status of the items.
It is very important to note that \emph{protecting GT is different from protecting the communication between the parties}.
To protect GT, one should make sure that information about the status of individual items is not revealed if a fraction of the test outcomes leaks. However, in GT, it is unlikely to assume that always one entity has access to all pool-tests, and can apply some encoding function before they are exposed. It is also unlikely to assume a mixer can add a certain substance that will prevent a third party from testing the sample. To protect GT, one should make sure that without altering mixed samples, if a fraction of them leaks, either already tested or not, information is not revealed.

While the current literature includes several works on the privacy in GT algorithms for digital objects \cite{atallah2008private,goodrich2005indexing,freedman2004efficient,rachlin2008secrecy}, these works are based on cryptographic schemes, assume the testing matrix is not known to all parties, impose a high computational burden, and, last but not least, assume the computational power of the eavesdropper is limited \cite{tagkey2004391,C13}. \emph{Information theoretic security} considered for secure communication \cite{C2,C13}, on the other hand, if applied appropriately to GT, can offer privacy at the price of \emph{additional tests}, without keys, obfuscation or assumptions on limited power. Due to the analogy between channel coding and group-testing regardless of security constraints, \cite{atia2012boolean,baldassini2013capacity}, in the supplementary materials of this work, we present an extensive survey of the literature on secure communication as well.
\subsection*{Main Contribution}
In this work, we formally define Secure Group Testing (SGT), suggest SGT algorithms based on information-theoretic principles and analyse their performance. In the considered model, there is an eavesdropper Eve who might observe part of the vector of pool-tests outcomes. The goal of the test designer is to design the tests in a manner such that a legitimate decoder can decode the status of the items (whether the items are defective or not) with an arbitrarily small error probability. It should {\it also} be the case that as long as Eve the eavesdropper gains only part of the output vector (a fraction $\delta$ - a bound on the value of $\delta$ is known {\it a priori} to the test designer, but which specific items are observed is not), Eve cannot (asymptotically, as the number of items being tested grows without bound) gain any significant information on the status of any of the items.

\ifdouble
\begin{figure}
\centering
\begin{subfigure}[b]{0.24\textwidth}
\includegraphics[scale=0.7]{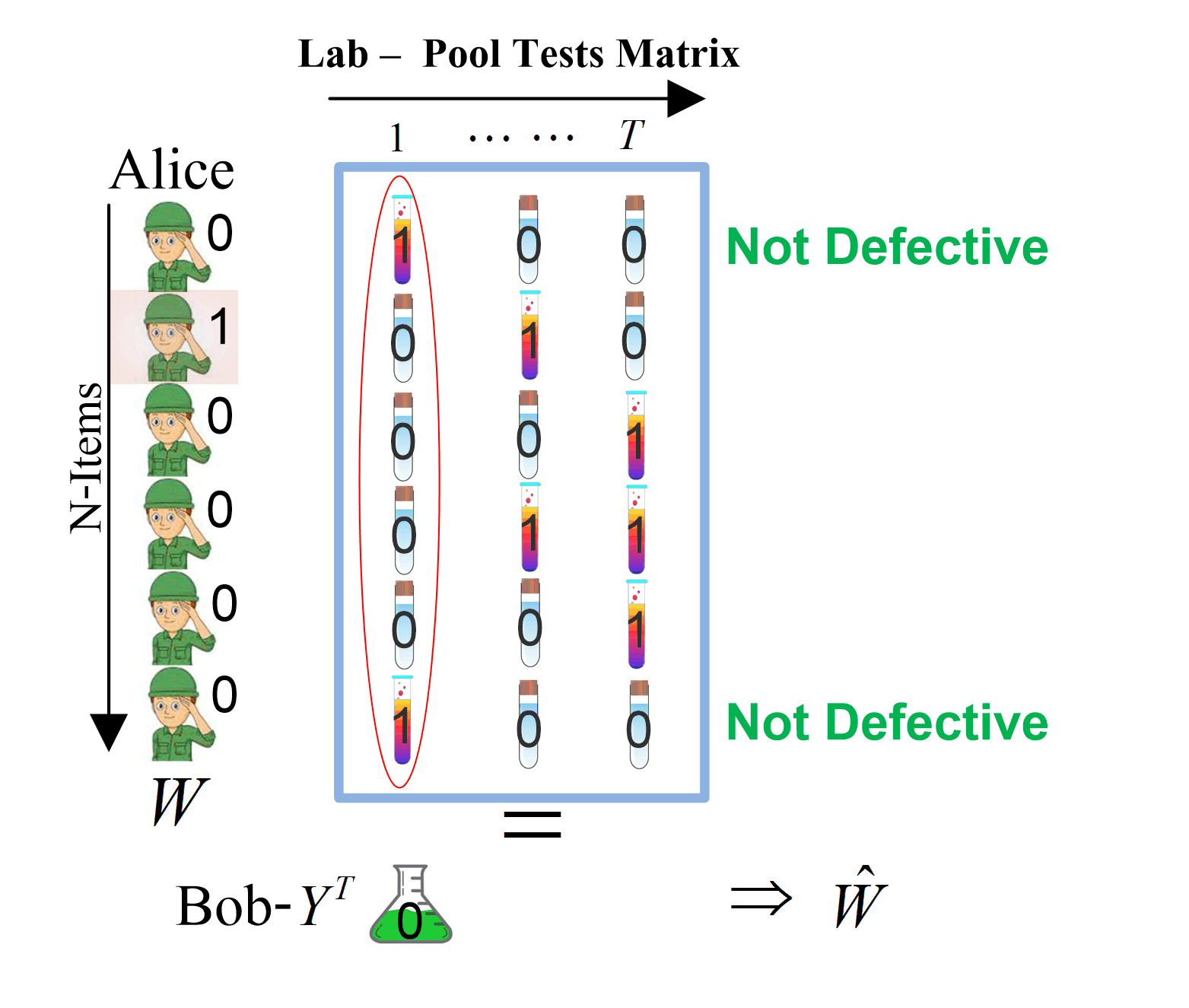}
\caption{}
\end{subfigure}
\begin{subfigure}[b]{0.24\textwidth}
\includegraphics[scale=0.7]{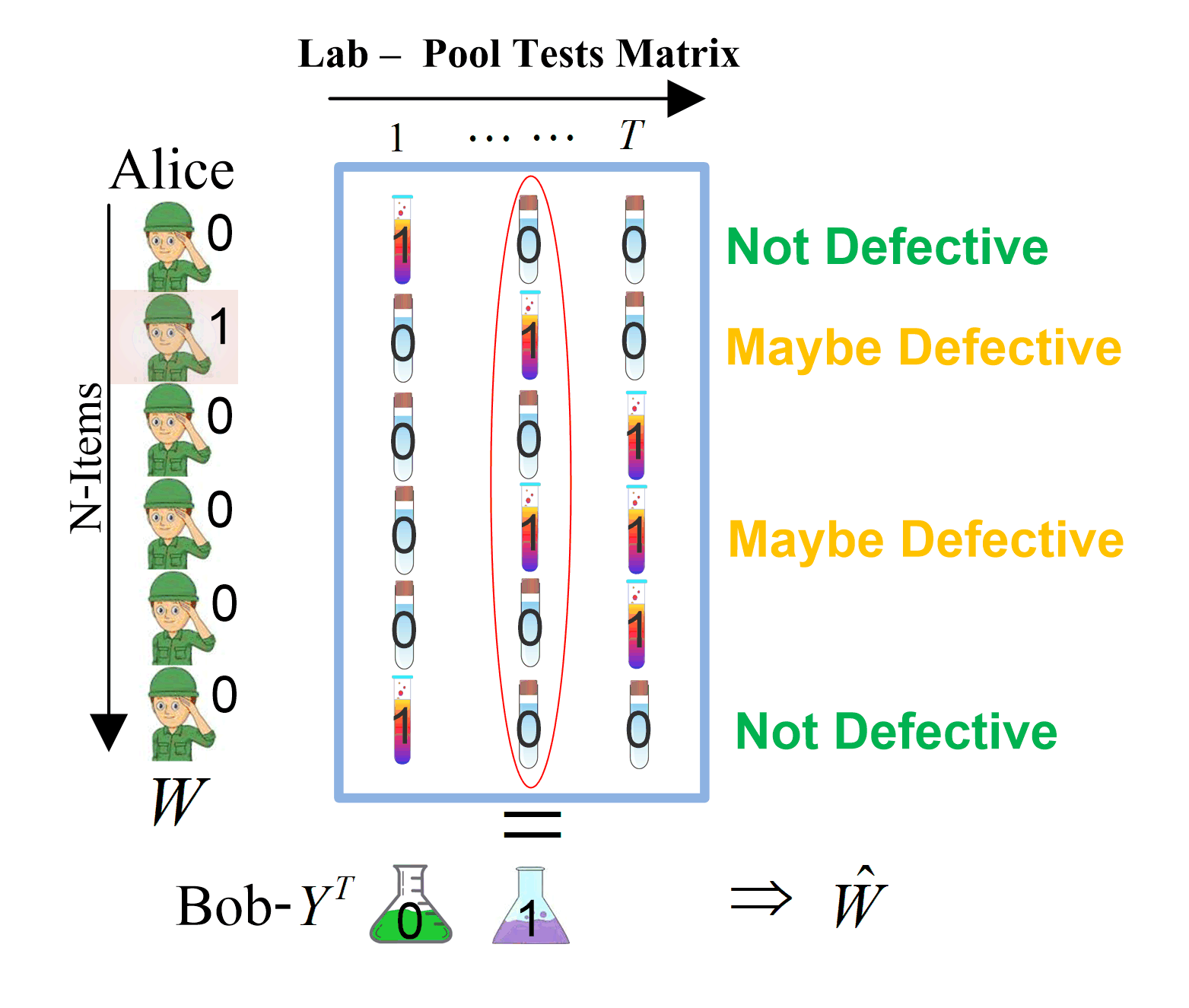}
\caption{}
\end{subfigure}
\begin{subfigure}[b]{0.24\textwidth}
\includegraphics[scale=0.7]{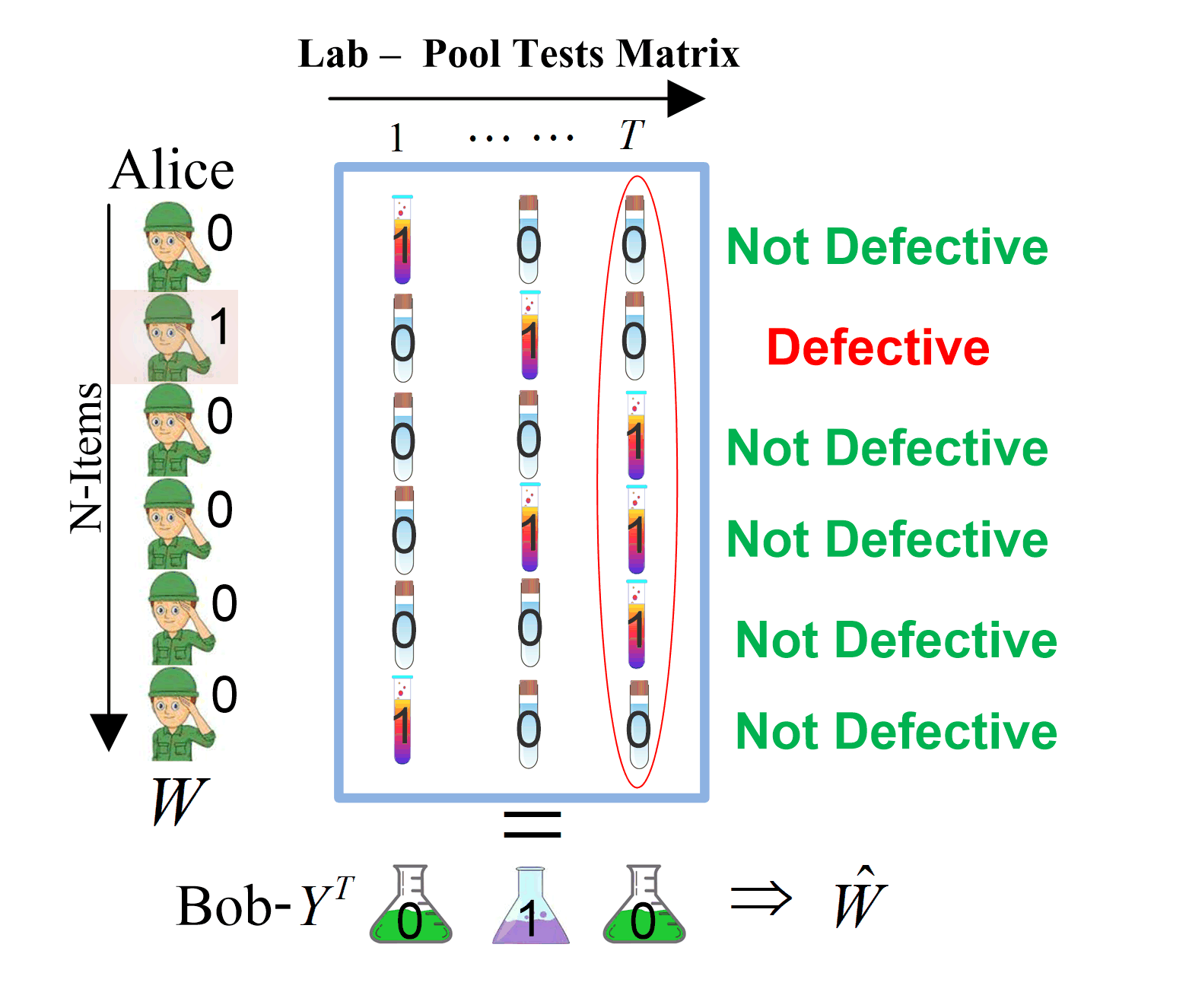}
\caption{}
\end{subfigure}
\begin{subfigure}[b]{0.24\textwidth}
\includegraphics[scale=0.6]{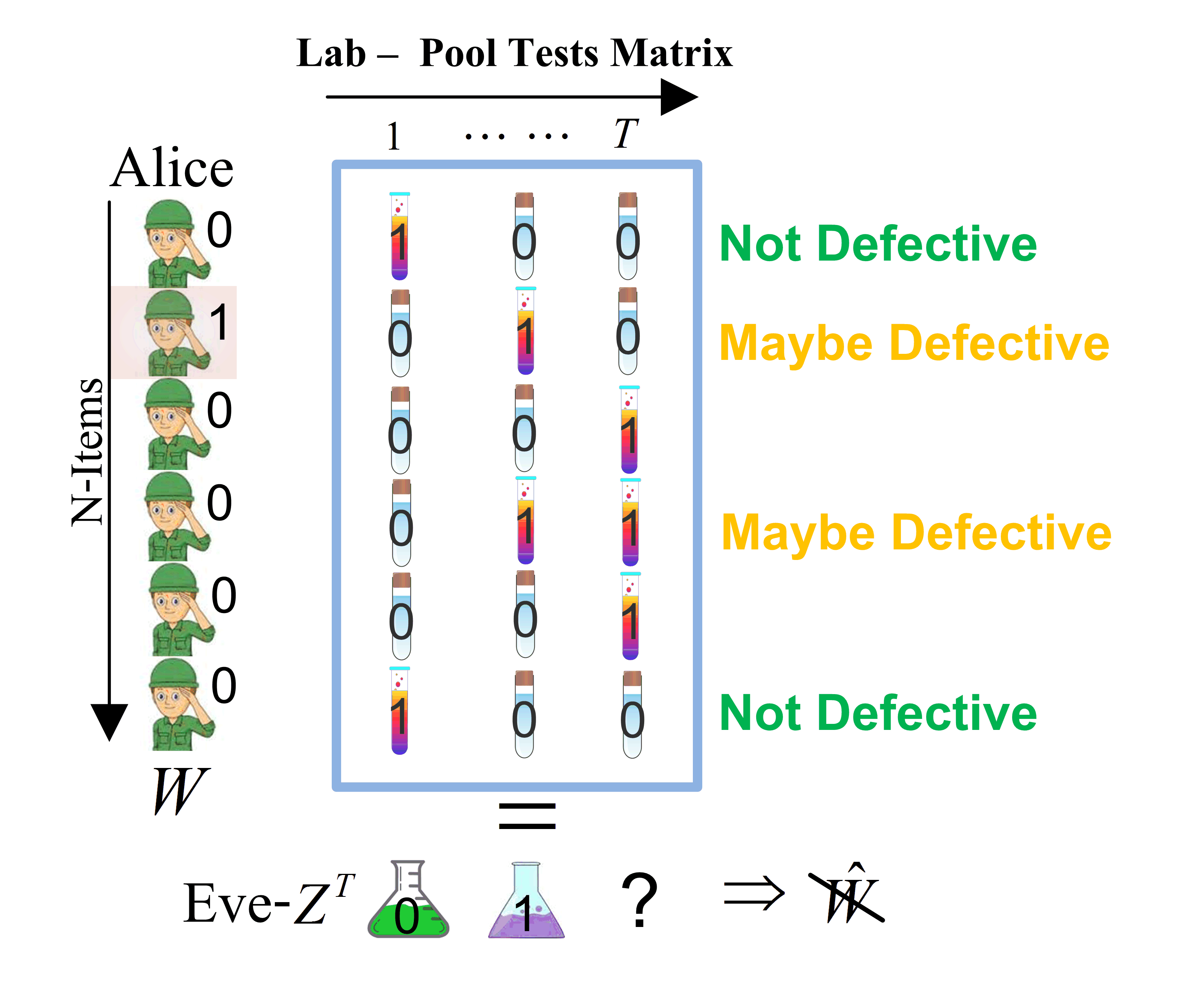}
\caption{}
\end{subfigure}
\caption[]{Classical group testing: \textit{An example of test results, a simple decoding procedure at the legitimate decoder and the risk of leakage. The example includes 7 items, out of which at most one defective (the second one in this case; unknown to the decoder). Three pooled tests are conducted. Each row dictates in which pooled tests the corresponding item participates. (a) Since the first result is negative, items 1 and 6 are not defective. (b) The second result is positive, hence at least one of items 2 and 4 is defective. (c) Based on the last result, as item 4 cannot be defective, it is clear that 2 is defective. Note that decoding in this case is simple: any algorithm which will simply rule out each item whose row in the matrix is not compatible with the result will rule out all but the second item, due to the first and last test results being negative, thus identifying the defective item easily. (d) An eavesdropper who has access to part of the results (the first two) can still infer useful information. Our goal is construct a testing matrix such that such an eavesdropper remains ignorant.}}
\label{basic_testing}
\vspace{-0.4cm}
\end{figure}
\else
\begin{figure}
\begin{flushleft}
\begin{subfigure}[b]{0.232\textwidth}
\includegraphics[scale=0.6]{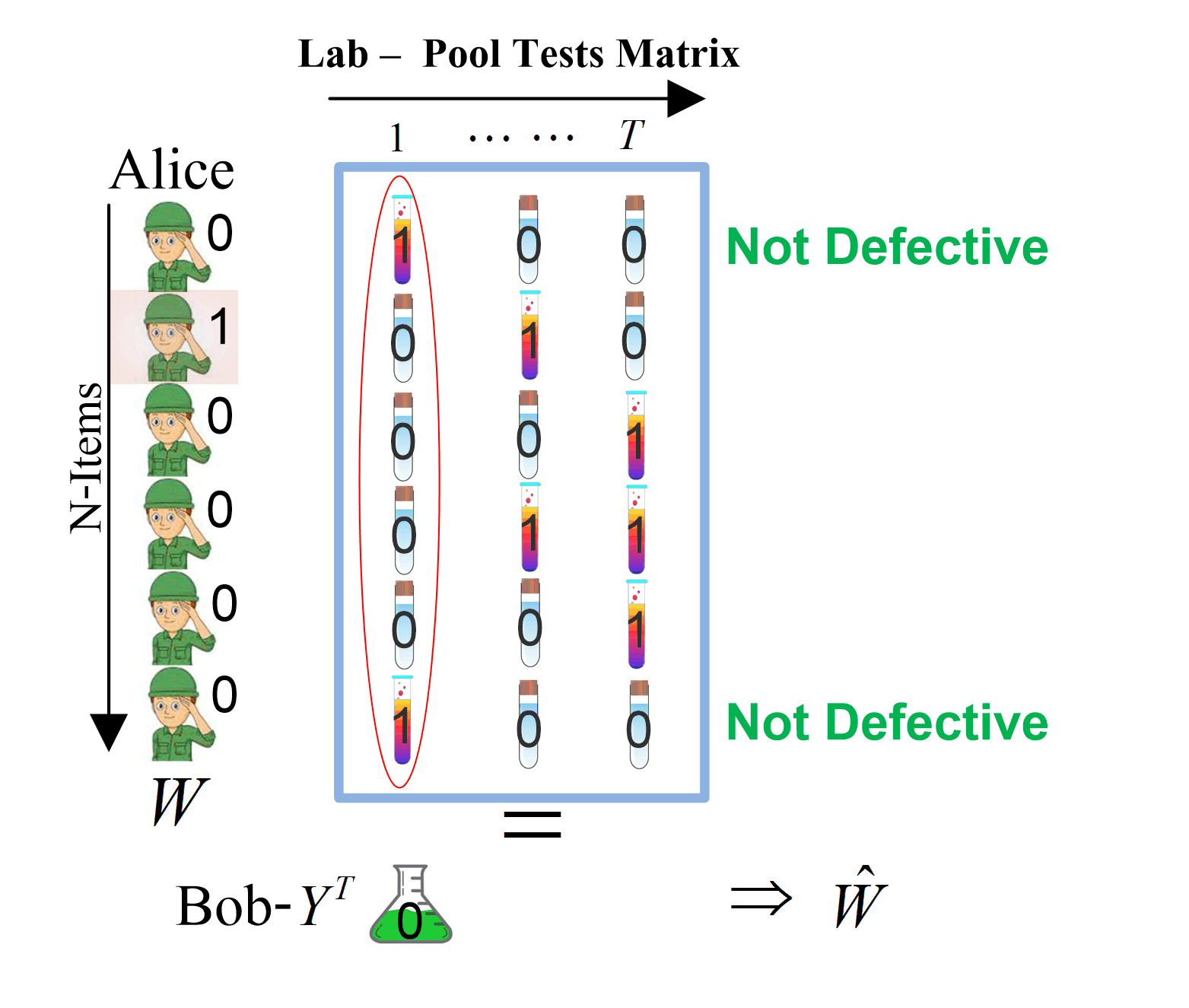}
\caption{}
\end{subfigure}
\begin{subfigure}[b]{0.232\textwidth}
\includegraphics[scale=0.6]{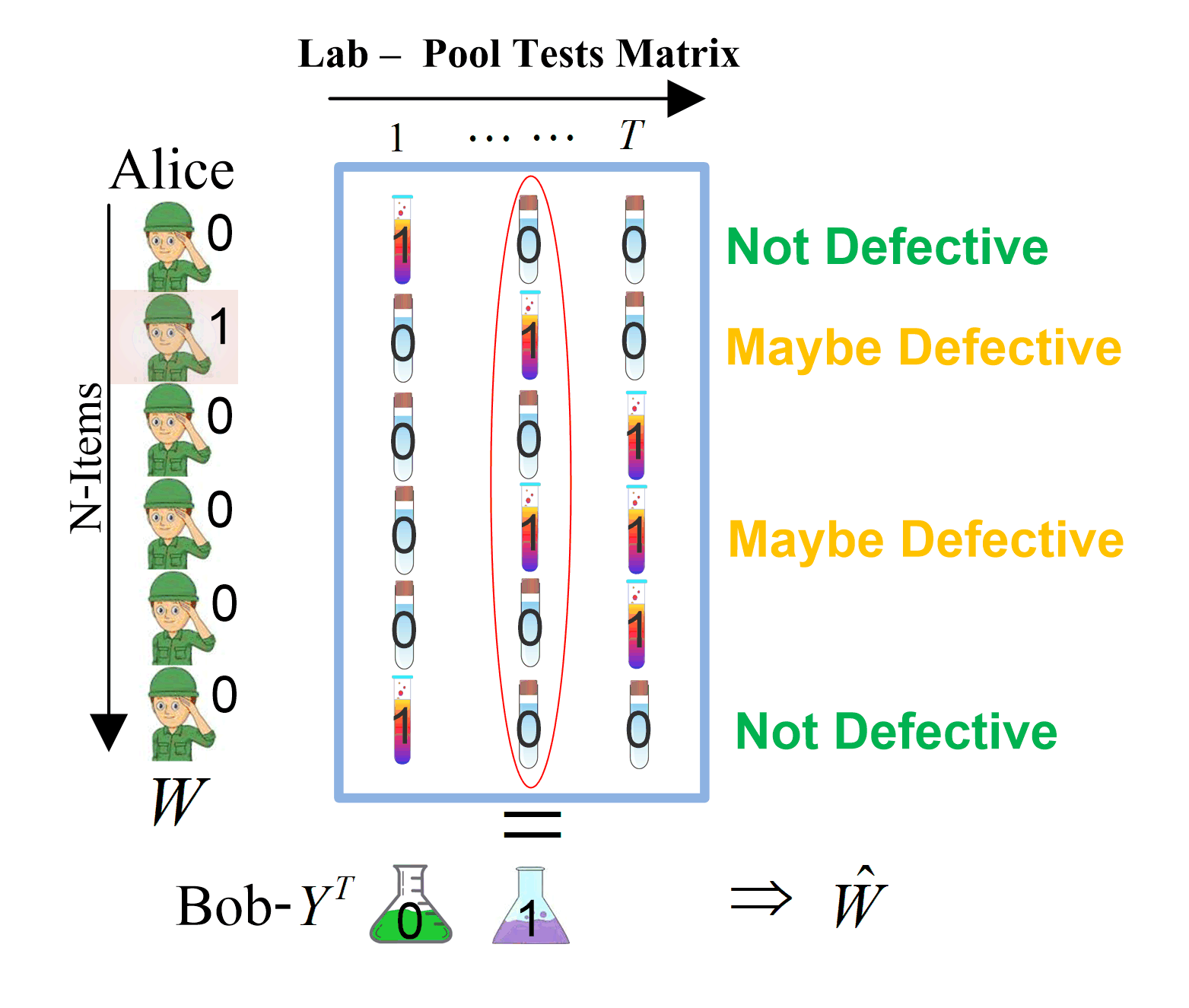}
\caption{}
\end{subfigure}
\begin{subfigure}[b]{0.232\textwidth}
\includegraphics[scale=0.6]{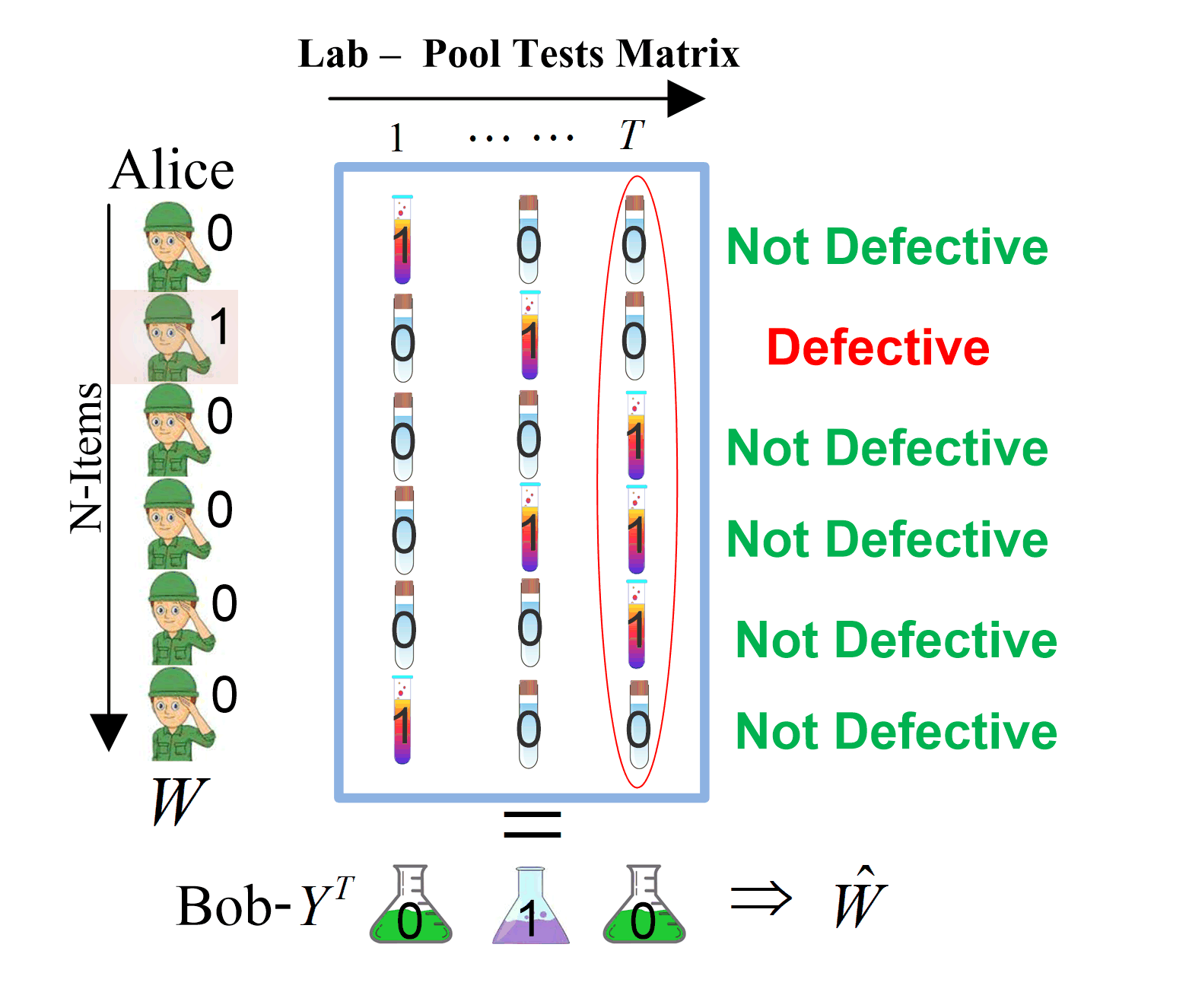}
\caption{}
\end{subfigure}
\begin{subfigure}[b]{0.232\textwidth}
\includegraphics[scale=0.6]{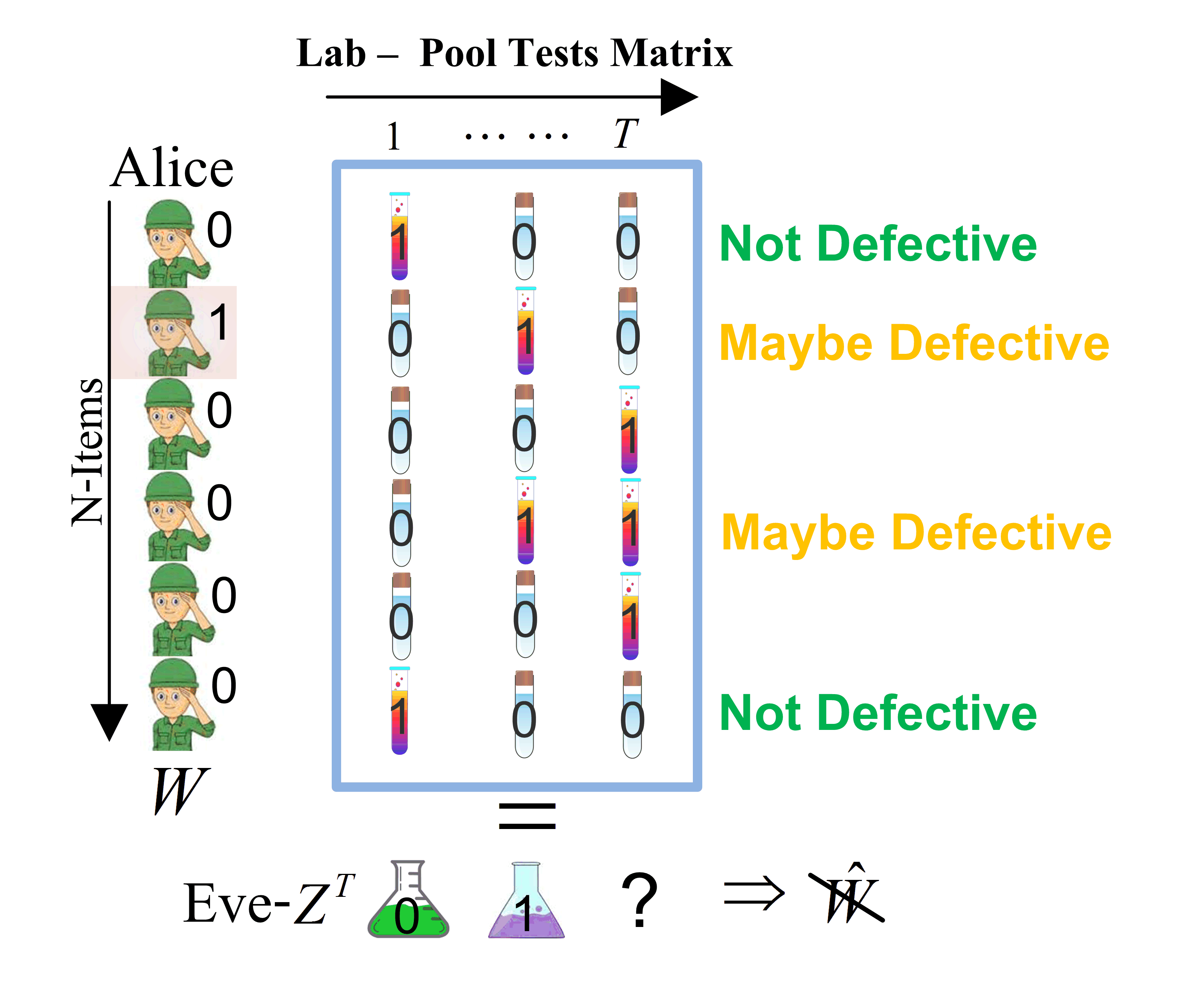}
\caption{}
\end{subfigure}
\end{flushleft}
\caption[]{Classical group testing: \textit{An example of test results, a simple decoding procedure at the legitimate decoder and the risk of leakage. The example includes 7 items, out of which at most one defective (the second one in this case; unknown to the decoder). Three pooled tests are conducted. Each row dictates in which pooled tests the corresponding item participates. (a) Since the first result is negative, items 1 and 6 are not defective. (b) The second result is positive, hence at least one of items 2 and 4 is defective. (c) Based on the last result, as item 4 cannot be defective, it is clear that 2 is defective. Note that decoding in this case is simple: any algorithm which will simply rule out each item whose row in the matrix is not compatible with the result will rule out all but the second item, due to the first and last test results being negative, thus identifying the defective item easily. (d) An eavesdropper who has access to part of the results (the first two) can still infer useful information. Our goal is construct a testing matrix such that such an eavesdropper remains ignorant.}}
\label{basic_testing}
\end{figure}
\fi

We propose a SGT code and corresponding decoding algorithms which ensure \tifs{high reliability} (with high probability over the test design, the legitimate decoder should be able to estimate the status of each item correctly), as well as \emph{\tifs{weak and} strong secrecy} conditions (as formally defined in \Cref{formulation}) - which ensures that essentially no information\footref{note2} about the status of individual items’ leaks to Eve.

Our first SGT code and corresponding decoding algorithm (based on Maximum Likelihood (ML) decoding) requires a number of tests that is essentially information-theoretically optimal in N, K and  $\delta$\tifs{, as demonstrated in \Cref{converse} by corresponding information-theoretic converse that we also show for the problem. The converse result shows that even guaranteeing weak security requires at least a certain number of tests, and the direct results show that essentially the same number of tests suffices to guarantee strong security. Hence both converse and direct results are with regard to the corresponding “harder to prove” notion of security.} 

The second code and corresponding decoding algorithm, while requiring a constant factor larger number of tests than is information-theoretically necessary (by a factor that is a function of $\delta$), is computationally efficient. It maintains the reliability and secrecy guarantees, yet requires only O($N^2T$) decoding time, where $T$ is the number of tests.

We do so by proposing a model, which is, in a sense, analogous to a {\it wiretap channel model}, as depicted in \Cref{figure:group_testing_model}.
In this analogy the subset of defective items (unknown {\it a priori} to all parties) takes the place of a confidential message. The testing matrix (representing the design of the pools - each row corresponds to the tests participated in by an item, and each column corresponds to a potential test) is a succinct representation of the encoder's codebook. \tifs{Rows of} this testing matrix can be considered as codewords. The decoding algorithm is analogous to a channel decoding process, and the eavesdropped signal is the output of \textit{an erasure channel}, namely, having only any part of the transmitted signal from the legitimate source to the legitimate receiver.

In classical non-adaptive group-testing, each row of the testing matrix comprises of a length-$T$ binary vector which determines which pool-tests the item is tested in.
In the SGT code constructions proposed in this work, each item instead corresponds to a vector chosen uniformly at random from a pre-specified set of random and independent vectors.
Namely, we use \emph{stochastic encoding}, and each vector corresponds to different sets of pool-tests an item may participate in. For each item the lab picks one of the vectors in its set (we term the set associated with item $j$ as ``Bin j") uniformly at random, and the item participates in the pool-tests according to this randomly chosen vector.
The set (``Bin") is known {\it a priori} to all parties, but the specific vector chosen by the encoder/mixer is only known to the encoder/mixer, and hence is {\it not a shared key/common randomness} in any sense. A schematic description of our procedure is depicted in \Cref{fig:WiretapCoding}.
\ifdouble
\begin{figure}
  \centering
  \includegraphics[trim=9cm 9.2cm 1cm 9.05cm,clip,scale=0.85]{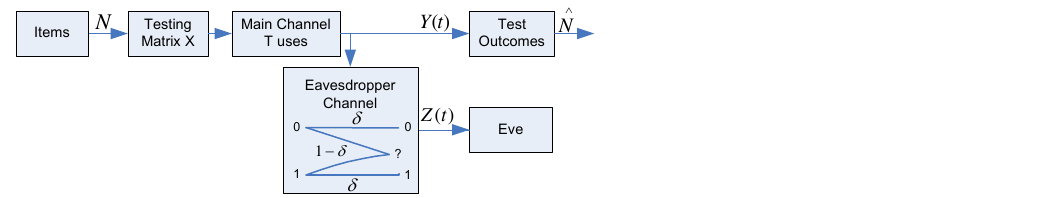}
  \caption{An analogy between a wiretap erasure channel and the corresponding SGT model.}
  \label{figure:group_testing_model}
  \vspace{-0.4cm}
\end{figure}
\else
\begin{figure}
  \centering
  \includegraphics[trim=6cm 9.2cm 0cm 9.05cm,clip,scale=1]{SecureGroupTestModelE.pdf}
  \caption{An analogy between a wiretap erasure channel and the corresponding SGT model.}
  \label{figure:group_testing_model}
\end{figure}
\fi

Accordingly, by obtaining a pool-test result, without knowing the specific vectors chosen by the lab for each item, the eavesdropper may gain only negligible information\footnote{\label{note2}This notion will be made precise for both weak and strong security constraints in \Cref{def_r_w_s}, with proofs in \tifs{\Cref{LowerBoundLeakage}} and Appendix \ref{strong_secrecy}, respectively.} regarding the items themselves. Specifically, we show that by careful design of the testing procedure, even though the pool-tests in which each item participated are chosen randomly and even though the legitimate user does not know \emph{a-priori} in which pool-tests each item has participated, the legitimate user will, with high probability over the testing procedure, be able to correctly identify the set of defective items, while the eavesdropper, observing only a subset of the pool-test results, will have no significant information regarding the status of the items.

To simplify the technical aspects and focus on the key methods, in this paper we consider an erasure channel at the eavesdropper, with erasure probability of $1-\delta$. Yet, in \Cref{noisy_eve}, we show that it is possible to generalize the main results given herein to the case where the outcome signal at the eavesdropper is affected by other noise models, e.g., false positive errors, false negative errors, both possible errors together or even a Binary Symmetric Channel (BSC).

Finally, we propose a few applications where the suggested secure GT coding scheme can be useful. Specifically,  we consider blood testing, anomaly detection in network data streams, and data aggregation and neighbor discovery in Wireless Sensor Networks (WSN).

The structure of this work is as follows. In \Cref{formulation}, a SGT model is formally described. \Cref{main results} includes our main results, with the direct proved in \Cref{LowerBound} and converse proved in \Cref{converse}. \Cref{noisy_eve} includes the case where different models of noisy observations are obtained at the eavesdropper. \Cref{efficient_algorithms} describes a computationally efficient algorithm, and proves an upper bound on its error probability. Strong security is proved in \Cref{strong_secrecy}.  In the supplementary materials of this work, we show a few examples, for which the SGT coding scheme is applicable, and an extensive survey.  %
\section{Problem Formulation}\label{formulation}
In \emph{SGT}, a legitimate user desires to identify a small unknown subset $\mathcal{K}$ of defective items from a larger set $\mathcal{N}$, while minimizing the number of measurements $T$ and keeping the eavesdropper, which is able to observe a subset of the tests results, ignorant regarding the status of the $\mathcal{N}$ items.
Let $N=|\mathcal{N}|$, $K=|\mathcal{K}|$ denote the total number of items, and the number of defective items, respectively.
As formally defined below, the legitimate user should (with high probability) be able to correctly estimate the set $\mathcal{K}$; on the other hand, from the eavesdropper’s perspective, this set should be ``almost" uniformly distributed over all possible ${N}\choose{K}$ sets. We assume that the number $K$ of defective items in $\mathcal{K}$ is known {\it a priori} to all parties - this is a common assumption in the GT literature \cite{macula1999probabilistic}.\footnote{If this is not the case, \cite{damaschke2010competitive,damaschke2010bounds} give methods/bounds on how to ``probably approximately" correctly learn the value of $K$ in a single stage with $O(\log N)$ tests.}

Throughout the paper, we use boldface to denote matrices, capital letters to denote random variables, lower case letters to denote their realizations, and calligraphic letters to denote the alphabet. Logarithms are in base $2$ and \tifs{$H(X)$ denotes the entropy of $X$. For a binary random variable with distribution $(p,1-p)$, with a slight abuse of notation, $H(p)$ denotes its entropy, namely, $-p\log p -(1-p)\log (1-p)$.}

\Cref{figure:secure-group-testing} gives a graphical representation of the model.
In general, and regardless of security constraints, non-adaptive GT is defined by a testing matrix
\begin{equation*}
\textbf{X}=[X_{1}^{T};X_{2}^{T};\ldots; X_{N}^{T}] \in \{0,1\}^{N\times T},
\end{equation*}
where each row corresponds to a separate item $j\in\{1,\ldots,N\}$, and each column corresponds to a separate pool test $t\in\{1,\ldots,T\}$.
For the $j$-th item,
\begin{equation*}
X_{j}^{T}=\{X_{j}(1),\ldots, X_{j}(T)\}
\end{equation*}
is a binary row vector, with the $t$-th entry $X_{j}(t)=1$ if and only if item $j$ participates in the $t$-th test.
\ifdouble
\begin{figure*}
  \centering
  \includegraphics[trim=0cm 0cm 0cm 0cm,clip,scale=0.8]{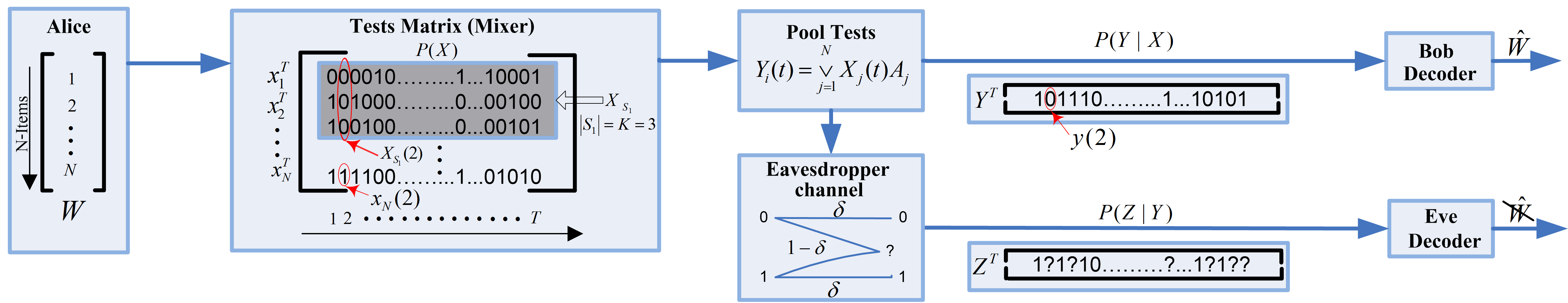}
  \caption{Noiseless non-adaptive secure group-testing setup.}
  \label{figure:secure-group-testing}
  \vspace{-0.4cm}
\end{figure*}
\else
\begin{figure*}
  \centering
  \includegraphics[trim=0cm 0cm 0cm 0cm,clip,scale=0.78]{secure-group-testing4.png}
  \caption{Noiseless non-adaptive secure group-testing setup.}
  \label{figure:secure-group-testing}
\end{figure*}
\fi
If $A_j\in \{0,1\}$ denotes an indicator function for the $j$-th item, determining whether it belongs to the defective set, i.e., $A_j=1$ if $j \in \mathcal{K}$ and $A_j = 0$ otherwise, the (binary) outcome of the $t \in \{1,\ldots,T\}$ pool test $Y(t)$ equals
\begin{equation*}
Y(t)=\bigvee_{j=1}^{N}X_{j}(t)A_j=\bigvee_{d\in \mathcal{K}}X_{d}(t),
\end{equation*}
where $\bigvee$ is used to denote the boolean OR operation.

In SGT, we assume an eavesdropper who observes a noisy vector $Z^{T}=\{Z(1),\ldots, Z(T)\}$, generated from the outcome vector $Y^{T}$. In the erasure case considered during this work, the probability of erasure is $1-\delta$, i.i.d.\ for each test.\footnote{In \Cref{noisy_eve}, we consider additional cases where Eve observes a general noisy vector $Z^T$, which is generated from $Y^T$ using different noise models, not necessarily a BEC.}
That is, on average, $T\delta$ outcomes are not erased and are accessible to the eavesdropper via $Z^{T}$. Therefore, in the erasure case, if $B_t\in \{1,?\}$ is an erasure indicator function for the $t$-th pool test, i.e., $B_t=1$ with probability $\delta$, and $B_t=?$ with probability $1-\delta$, the eavesdropper observes
\begin{equation*}
Z(t)=Y(t)B_t= \left(\bigvee_{j=1}^{N}X_{j}(t)A_j\right)B_t,\text{ }\text{ } t=1,\ldots,T.
\end{equation*}

Denote by $W \in \mathcal{W} \triangleq \{1,\ldots, {N \choose K}\}$ the \tifs{random} index of the \emph{subset of defective items}\tifs{, and by $w$ a specific realization of the subset}. We assume $W$ is uniformly distributed, that is, there is no \emph{a priori} bias towards any specific subset.\footnote{\tifs{This is a common model considered in the GT literature, in which we assume $K$ defectives items. $K$ is fixed and known. The choice of which $K$ out of $N$ is random and uniform, and unknown to all the parties. Another model considered in the GT literature assumes that each item is defective with some probability, i.i.d. across items.} \tifs{In many group-testing scenarios one can relate results for the second model presented above, to other scenarios (e.g., fixed $K$), hence we focus on the first model, where {\it exactly} $K$ items are defective.}}
Further, denote by $\hat{W}(Y^T)$ the index recovered by the legitimate decoder, after observing $Y^T$. In this work, we assume that the mixer may use a \emph{randomized} testing matrix. In this case, the random bits used \tifs{are known} only to the mixer, and are not assumed to be shared with the decoder. In other words, the ``codebook" which consists of all possible testing matrices is known to all parties, Alice, Bob and Eve. However, if the mixer \tifs{choosing particular} $\textbf{X}$, the random value is not shared with Bob or Eve. We refer to the codebook consisting of all possible matrices, together with the decoder at Bob's side as SGT algorithm.

As we are interested in the asymptotic behavior, i.e., in ``capacity style" results, with a focus on the number of tests $T$ (as a function of $N$ and $K$) required to guarantee high probability of recovery as the number of items $N$ grows without bound. For simplicity, in the first part of this work, we focus primarily on the regime where $K$ is a constant independent of $N$. In \Cref{efficient_algorithms}, we give an algorithm which applies to any $K$.\footnote{Following the lead of~\cite{aldridge2017capacity}, in principle, many of our results in this section as well can be extended to the regime where $K = o(N^{1/3})$, but for ease of presentation we do not do so here).} The following definition lays out the goals of SGT algorithms.
\begin{definition}\label{def_r_w_s}
A sequence of SGT algorithms with parameters $N,K$ and $T$ is asymptotically (in N) \emph{reliable} and \emph{weakly} or \emph{strongly} secure if,

\noindent (1) \emph{Reliable}: The probability (over the index $W$) of incorrect reconstruction of $W$ at the legitimate receiver converges to zero. That is,
\begin{equation*}
\lim_{N \to \infty} P(\hat{W}(Y^T) \ne W) = 0.
\end{equation*}

\noindent (2) \emph{Weakly secure}: One potential security goal is so-called {\it weak information-theoretic security} against eavesdropping. Specifically, if the eavesdropper observes $Z^T$, a scheme is said to be weakly secure if
\begin{equation*}
\lim_{T \to \infty} \frac{1}{T}I(W;Z^T) = 0.
\end{equation*}
(3) \emph{Strongly secure}: A stronger notion of security is so-called {\it strong information-theoretic security} against eavesdropping. Specifically, if the eavesdropper observes $Z^T$, a scheme is said to be strongly secure if
\begin{equation*}
\lim_{T \to \infty}I(W;Z^T) = 0.
\end{equation*}
\end{definition}
\begin{remark}
  Note that strong security implies that in the limit the distribution over $Z^T$ is essentially statistically independent\footref{note2} of the distribution over $W$. Specifically, the {\it KL divergence} between $p_{Z^T,W}$ and $p_{Z^T}p_{W}$ converges to $0$.
\end{remark}
\begin{remark}\label{whyW}
  While weak security is a much weaker notation of security against eavesdropping than strong security, and indeed is implied by strong security, nonetheless we consider it in this work for the following reason. Our impossibility result will show that even guaranteeing weak security requires at least a certain number of tests, and our achievability results will show that essentially the same number of tests suffices to guarantee strong security. Hence both our impossibility and achievability results are with regard to the corresponding ``harder to prove" notion of security.
\end{remark}
To conclude, the goal in this work is to design (for parameters $N$ and $K$) an $N \times T$ measurement matrix (which is possibly randomized) and a decoding algorithm \tifs{$\hat{W}(Y^T)$}, such that on observing $Y^T$, the legitimate decoder can (with high probability over $W$) identify the subset of defective items, and yet, on observing $Z^T$, the eavesdropper learns essentially nothing\footref{note2} about the set of defective items.  %
\section{Main Results}\label{main results}
Under the model definition given in \Cref{formulation}, our main results are the following sufficiency (direct) and necessity (converse) conditions, characterizing the maximal number of tests required to guarantee both reliability and security. The proofs are deferred to \Cref{LowerBound}, \Cref{converse} and \Cref{strong_secrecy}.
\subsection{Direct (Sufficiency)}
The sufficiency part is given by the following theorem.
\begin{theorem}\label{direct theorem1}
Assume a SGT model with $N$ items, out of which $K=O(1)$ are defective. For any $0 \leq \delta < 1$, if
\begin{equation}\label{main_result_eq}
 T \geq \frac{1 + \varepsilon}{1 - \delta} K \log N,
\end{equation}
for some $\varepsilon > 0$ independent of $N$ and $K$, then there exists a sequence of SGT algorithms which are reliable and secure. That is, as $N\rightarrow \infty$, both the average error probability approaches zero exponentially and an eavesdropper with leakage probability $\delta$ is kept ignorant\tifs{, such that $\frac{1}{T}I(W;Z^T) \leq \epsilon_T + \epsilon^{\prime}$ where $\epsilon_T \to 0$ as $T \to \infty$.}
\end{theorem}
The construction of the SGT algorithm, together with the proofs of reliability and weak secrecy are deferred to Section \ref{LowerBound}. The strong secrecy proof\tifs{, under which $\lim_{T \to \infty}I(W;Z^T) = 0$,} is deferred to \Cref{strong_secrecy}.
In fact, in \Cref{LowerBound} we actually prove that the error probability decays to $0$. However, a few important remarks are in order now.

First, it is important to note that compared to only a reliability constraint, the number of tests required for \emph{both reliability and secrecy} is increased by the multiplicative factor $\frac{1}{1-\delta}$, where, again, $\delta$ is the leakage probability at the eavesdropper. Together with the converse below, this suggests
\[
T= \Theta \left(\frac{K\log N}{1-\delta}\right),
\]
and, a $\Theta \left(K\log N\right)$ result for $\delta$ bounded away from $1$.

The result given in \Cref{direct theorem1} uses an ML decoding at the legitimate receiver. The complexity burden in ML, however, prohibits the use of this result for large N. In \Cref{direct theorem5}, we suggest an efficient decoding algorithm, which maintains the reliability and the secrecy results using a much simpler decoding rule, at the price of only slightly more tests.

The results thus far were for $K=O(1)$. Clearly, it is interesting to analyse the regime where $K$ is allowed to grow with $N$, that is $K=O(\log(N))$. This regime is indeed analyzed in \Cref{Kgrow}. The asymptotic result is given by the Theorem below.
\begin{theorem}\label{theorem_K=o(N)}
Assume an SGT model with $N$ items, out of which $K=O(\log(N))$ are defective. For any $0 \leq \delta < 1$, reliability and secrecy, can be maintained if
\begin{equation*}
T =  \Theta \Bigg(\frac{K\log N}{1-\delta}\Bigg),
\end{equation*}
\tifs{that is $\frac{1}{T}I(W;Z^T) \leq \epsilon_T + \epsilon^{\prime}$ where $\epsilon_T \to 0$ as $T \to \infty$.}
\end{theorem}
\begin{remark}\label{OP}
  While order-optimal result are obtained for the regime of $K=O(\log(N))$, following the analysis given in \Cref{Kgrow}, in the regime of $K=\omega(\log(N))$, reliability and secrecy in the SGT model can still be achieved for any $0 \leq \delta < 1$, yet the bound is no longer tight. For example, when $K = N^{1/4}$, the reliability and secrecy constraints are satisfied for $T=O\left(\frac{\sqrt N\log(N)}{1-\delta}\right)$.
\end{remark}
\subsection{Converse (Necessity)}
The necessity part is given by the following theorem.
\begin{theorem}\label{converse theorem}
Let $T$ be the minimum number of tests necessary to identify a defective set of \tifs{any} cardinality $K$ among population of size $N$ while keeping an eavesdropper, with a leakage probability $\delta < 1$, ignorant regarding the status of the items. Then, if $\frac{1}{T} I(W;Z^T)<\epsilon$, one must have:
\begin{eqnarray*}
T \ge \frac{1-\epsilon_T}{1-\delta}\log\binom{N}{K},
\end{eqnarray*}
where $\epsilon_T = \epsilon +\tilde{\epsilon}_T$, with $\tilde{\epsilon}_T\rightarrow 0$ as $T\rightarrow \infty$.
\end{theorem}
The lower bound is derived using Fano's inequality to address reliability, assuming a negligible mutual information at the eavesdropper, thus keeping an eavesdropper with leakage probability $\delta$ ignorant, and information inequalities bounding the rate of the message on the one hand, and the data Eve \emph{does not see} \tifs{on the other hand}. Compared with the lower bound without security constraints, it is increased by the multiplicative factor $\frac{1}{1-\delta}$.
\subsection{Secrecy capacity in SGT}
Returning to the analogy in \cite{baldassini2013capacity} between channel capacity and group testing, one might define by $C_s$ the (asymptotic) minimal threshold value for $\log\binom{N}{K}/T$, above which no reliable and secure scheme is possible. Under this definition, the result in this paper show that $C_s = (1-\delta)C $, where $C$ is the capacity without the security constraint. Clearly, this can be written as
\begin{equation*}
C_s = C - \delta C,
\end{equation*}
raising the usual interpretation as the \emph{difference} between the capacity to the legitimate decoder and that to the eavesdropper \cite{C13}. Note that as the effective number of tests Eve sees is $T_e = \delta T$, her GT capacity is $\delta C$.
\subsection{Efficient Algorithms}
Under the SGT model definition given in \Cref{formulation}, we further consider a computationally efficient algorithm at the legitimate decoder. \tifs{Specifically, we analyze the \emph{Definitely Non-Defective} (DND) algorithm (also called \emph{Combinatorial Orthogonal Matching Pursuit} (COMP)), considered for the non-secure GT model in the literature \cite{chan2014non,aldridge2014group}.} The idea for this efficient decoding algorithm was presented first in \cite{kautz1964nonrandom}, and later in several variants \cite{chen2008survey,chan2011non,lo2013efficient,malyutov2013search,aldridge2019group}. The theorem below states that indeed efficient decoding (with arbitrarily small error probability) and secrecy are possible, at the price of a higher $T$. Interestingly, the theorem applies to any $K$, and not necessarily only to $K = O(1)$. This is, on top of the reduced complexity, an important benefit of the suggested algorithm.
\begin{theorem}\label{direct theorem5}
Assume an SGT model with $N$ items, out of which \tifs{any} $K$ are defective. Then, for any $\delta < \frac{1}{2}\left(1-\frac{\ln 2}{K}\right)$, there exists an \tifs{DND} efficient decoding algorithm, requiring $O(N^2T)$ operations, such that \tifs{with $\frac{1}{T}I(W;Z^T) \leq \epsilon_T + \epsilon^{\prime}$ where $\epsilon_T \to 0$ as $T \to \infty$.}, if the number of tests satisfies
\[
T \ge \frac{1+\epsilon}{\frac{1}{2}(1-\frac{\ln 2}{K})-\delta} K \log N
\]
its error probability is upper bounded by
\[
P_e \leq N^{-\epsilon}.
\]
\end{theorem}
The construction of the DND GT algorithm, together with the proofs of reliability and secrecy are deferred to \Cref{efficient_algorithms}. Clearly, the benefits of the algorithm above come at the price of additional tests and a smaller range of $\delta$ it can handle.

\tifs{The results in \Cref{direct theorem1} using ML decoding show that any value of $\delta <1$ is possible (with a $\frac{1}{1-\delta}$ toll on $T$ compared to non-secure GT). The results in \Cref{direct theorem5} suggest that using the DND efficient algorithm, one can have a small error probability only for $\delta < 1/2$, and the toll on $T$ is greater than $\frac{1}{\frac{1}{2}-\delta}$. This is consistent with the fact that this algorithm is known to achieve only slightly more than half of the capacity for non-secure GT \cite{aldridge2014group}. Both these results may be due to an inherent deficiency in the decoding algorithm.}
\section{Code Construction and a Proof for \Cref{direct theorem1}} \label{LowerBound}
In order to keep the eavesdropper, which obtains only a fraction $\delta$ of the outcomes, ignorant regarding the status of the items, we \emph{randomly} map the items to the tests. Specifically, as depicted in \Cref{fig:WiretapCoding}, for each item we generate a bin, containing several rows. The number of such rows \emph{corresponds} to the number of tests that the eavesdropper can obtain, yet, unlike wiretap channels, \emph{it is not identical} to the number of outcomes that the eavesdropper can obtain, i.e., to Eve's capacity, and should be normalized by the number of defective items.
Then, for the $j$-th item, we randomly select a row from the $j$-th bin. This row will determine in which tests the item will participate.

In order to rigorously describe the construction of the matrices and bins, determine the exact values of the parameters (e.g., bin size), and analyze the reliability and secrecy, we first briefly review the representation of the GT problem as a channel coding problem \cite{atia2012boolean}, together with the components required for SGT.

An SGT code consists of an index set $\mathcal{W} =\{1,2,\ldots \binom{N}{K}\}$, its $w$-th item corresponding to the $w$-th subset $\mathcal{K}\subset \{1,\ldots,N\}$; A discrete memoryless source of randomness $(\mathcal{R},p_R)$, with known alphabet $\mathcal{R}$ and known statistics $p_R$; An encoder,
\begin{equation*}
f : \mathcal{W} \times \mathcal{R} \rightarrow \mathcal{X}_{S_w}\in\{0,1\}^{K\times T}
\end{equation*}
which maps the index $W$ of the defective items to a matrix $\textbf{X}^{T}_{S_{w}}$ of codewords, each of its rows corresponding to a different item in the index set $S_w$, $w\in \mathcal{W}$, $|S_w|=K$. Note that the notation $\textbf{X}^{T}_{S_{w}}$ depends on the actual randomness selected at the encoder. However, we do not include the dependence on $\mathcal{R}$ in the notation to keep this notation simple.
The need for a \emph{stochastic encoder} is similar to most encoders ensuring information theoretic security, as randomness is required to confuse the eavesdropper about the actual information \cite{C13}. Hence, we define by $R_K$ the local randomness variable at the mixer, encompassing the randomness required at the $K$ defective items. Further define $M$ as the number of rows in each bin. Thus, $\log(M^K)=H(R_K)$. Note that this computation refers only to the amount of randomness used for the $K$ actually defective items. 

At this point, an important clarification is in order. The lab, of course, does not know which items are defective. Thus, operationally, it needs to select a row for each item. However, in \emph{the analysis}, since only the defective items affect the output (that is, only their rows are ORed together to give $Y^T$), we refer to the ``message" as the index of the defective set $w$ and refer only to the random variable $R_K$ required to choose the rows in their bins.
In other words, unlike the analogous communication problem, in GT, \emph{nature} performs the actual mapping from $W$ to $\textbf{X}^{T}_{S_{w}}$. The mixer only mixes the blood samples according to the (random in this case) testing matrix it has.

A decoder at the legitimate user is a map
\begin{equation*}
\hat{W} : \mathcal{Y}^{T} \rightarrow \mathcal{W}.
\end{equation*}
The probability of error is $P(\hat{W}(Y^{T})\neq W)$. The probability that an outcome test leaks to the eavesdropper is $\delta$. We assume a memoryless model, i.e., each outcome $Y(t)$ depends only on the corresponding input $X_{S_w}(t)$, and the eavesdropper observes $Z(t)$, generated from $Y(t)$ according to
\begin{equation*}
p(Y^T,Z^T|X_{S_w})=\prod_{t=1}^{T}p(Y(t)|X_{S_w}(t))p(Z(t)|Y(t)).
\end{equation*}
We may now turn to the detailed construction and analysis.
\ifdouble
\begin{figure}
  \centering
  \includegraphics[trim= 9.3cm 8cm 2cm 7.6cm,clip,scale=0.9]{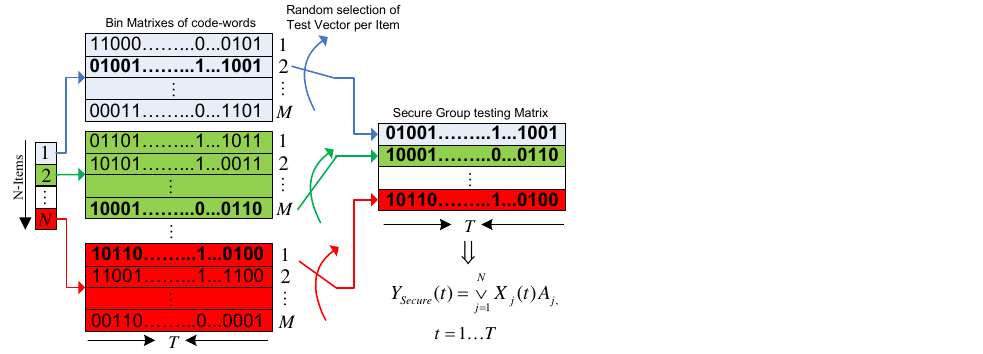}
  \caption{Binning and encoding process for a SGT code.}
  \label{fig:WiretapCoding}
  \vspace{-0.4cm}
\end{figure}
\else
\begin{figure}
  \centering
  \includegraphics[trim= 6.3cm 8cm 2cm 7.6cm,clip,scale=1]{Wiretap_coding5_one_col.pdf}
  \caption{Binning and encoding process for a SGT code.}
  \label{fig:WiretapCoding}
\end{figure}
\fi
\subsubsection{Codebook Generation}
Choose M such that $$\log_2(M) = T(\delta-\epsilon^{\prime})/\tifs{K}$$ for some $\epsilon^{\prime}>0$. $\epsilon^{\prime}$ will affect the equivocation.
Using a distribution $P(X^{T})=\prod^{T}_{i=1}P(x_i)$, for each item generate $M$ independent and identically distributed codewords $x^{T}_{j}(m)$, $1 \leq m \leq M$.
That is, specifically each codeword is generated randomly under a fixed Bernoulli distribution of $\left(p=\frac{\ln(2)}{K},q=1-\frac{\ln(2)}{K}\right)$\footnote{\label{note10}\tifs{The optimal Bernoulli parameter $p$ for the testing matrix is $1-2^{-1/K}$ which results in half positive pool-tests and half negative on average. To simplify the expressions, we approximate $1-2^{-1/K}$ by $\ln(2)/K$ \cite{aldridge2017almost}. Yet, note that it is possible to replace $\ln(2)/K$ by $1-2^{-1/K}$, or for preciseness carry the error term $O(1/K^2)$, which is negligible for large enough $K$.}}.
The codebook is depicted in the left hand side of \Cref{fig:WiretapCoding}.
Reveal the codebook to Alice and Bob. We assume Eve may have the codebook as well.
\subsubsection{Testing}
For each item $j$, the mixer/lab selects uniformly at random one codeword $x^{T}_{j}(m)$ from the $j$-th bin.
Therefore, the SGT matrix contains $N$ randomly selected codewords of length $T$, one for each item, defective or not. Amongst is an unknown subset $X^{T}_{S_{w}}$, with the index $w$ representing the true defective items. An entry of the $j$-th random codeword is $1$ if the $j$-item is a member of the designated pool test and $0$ otherwise.
\subsubsection{Decoding at the Legitimate Receiver}
The decoder looks for a collection of $K$ codewords $X_{S_{\hat{w}}}^{T}$, \textit{at most one from each bin}, for which $Y^T$ is most likely. Namely,
\begin{equation*}
P(Y^{T}|X_{S_{\hat{w}}}^{T})>P(Y^{T}|X_{S_{w}}^{T}), \forall w \neq \hat{w}.
\end{equation*}
Then, the legitimate user (Bob) declares $\hat{W}(Y^T)$ as the set of bins in which the rows $\hat{w}$ reside.
\subsection{Reliability}\label{relia}
Let $(\mathcal{S}^{1},\mathcal{S}^{2})$ denote a partition of the defective set $S$ into disjoint sets $\mathcal{S}^{1}$ and $\mathcal{S}^{2}$, with cardinalities $i$ and $K-i$, respectively.\footnote{This partition helps decompose the error events into classes,  where in class $i$ one already knows $K-i$ defective items, and the dominant error event corresponds to missing the other $i$. Thus, it is easier to present the error event as one ``codeword" against another.} \tifs{Recall, $X_{\mathcal{S}}$ is a column of the matrix $\textbf{X}^{T}_{S}$.}
Let $I(X_{\mathcal{S}^1};X_{\mathcal{S}^2},Y)$, with fixed non-random $(\mathcal{S}_1; \mathcal{S}_2)$, denote the mutual information between $X_{\mathcal{S}^1}$ and $X_{\mathcal{S}^2},Y$, under the i.i.d.\ distribution with which the codebook was generated and remembering that $Y$ is the output of a Boolean channel.
The following lemma is a key step in proving the reliability of the decoding algorithm.
\begin{lemma}\label{direct lemma1}
\tifs{For $K=O(1)$,} if the number of tests satisfies
\begin{eqnarray*}
T  \geq (1+\varepsilon)\cdot\max_{i=1,\ldots ,K} \frac{\log\binom{N-K}{i}M^i}{I(X_{\mathcal{S}^1};X_{\mathcal{S}^2},Y)},
\end{eqnarray*}
then, under the codebook above, as $N\rightarrow \infty$ the average error probability approaches zero.
\end{lemma}
\iftifs
\Cref{direct lemma1}, is an adaptation of the results in \cite{atia2012boolean} to the codebook required for SGT. Specifically, to obtain a bound on the required number of tests as given in \Cref{direct lemma1}, we first state Lemma 2, which bounds the error probability of the ML decoder using a Gallager-type bound \cite{gallager1968information}.
\else
\tifs{To prove \Cref{direct lemma1}, which extends the results in \cite{atia2012boolean} to the codebook required for SGT, namely, to obtain a bound on the required number of tests, we first state \Cref{error lemma2} below, which bounds the error probability of the ML decoder using a Gallager-type bound \cite{gallager1968information}. \Cref{error lemma2} with then be used to prove \Cref{direct lemma1}.}
\fi

\begin{definition}
\tifs{The} error event $E_{i}$ in the ML decoder is defined as the event of mistaking the true set for a set which differs from it in exactly $i$ items.
\end{definition}
\begin{lemma}\label{error lemma2}
The error probability $P(E_{i})$ is bounded by
\begin{equation*}
P(E_{i}) \leq 2^{-T\left(E_{o}(\rho)-\rho\frac{\log\binom{N-K}{i}M^i}{T}-\frac{\log\binom{K}{i}}{T}\ - \frac{K}{T}\right)},
\end{equation*}
where the error exponent $E_{o}(\rho)$ is given by
\ifdouble
\begin{multline}
E_{o}(\rho)= -\log \sum_{Y\in \{0,1\}}\sum_{X_{\mathcal{S}^2}}\Bigg[\sum_{X_{\mathcal{S}^1}}P(X_{\mathcal{S}^1})\\
p(Y,X_{\mathcal{S}^2}|X_{\mathcal{S}^1})^{\frac{1}{1+\rho}}\Bigg]^{1+\rho},\text{ } 0 \leq\rho \leq1.\nonumber
\end{multline}
\else
\begin{equation*}
E_{o}(\rho)= -\log \sum_{Y\in \{0,1\}}\sum_{X_{\mathcal{S}^2}}\Bigg[\sum_{X_{\mathcal{S}^1}}P(X_{\mathcal{S}^1})
p(Y,X_{\mathcal{S}^2}|X_{\mathcal{S}^1})^{\frac{1}{1+\rho}}\Bigg]^{1+\rho},\text{ } 0 \leq\rho \leq1.
\end{equation*}
\fi
\end{lemma}
\tifs{The proof of \Cref{error lemma2} is given in Appendix \ref{appendix:APPENDIX A}. It is based on \cite[Lemma III.1]{atia2012boolean}, which studied non-secure group testing, however, as shown in Appendix \ref{appendix:APPENDIX A} for the secure setting, there is a key difference in the analysis of the error probability.}

Note that unlike the typical encoders ensuring information theoretic security \cite{C13}, for which the exponent of the size of each bin equals Eve's capacity, in the code suggested the number of the codewords in each bin is normalized by the number of defective items $K$. This reduces the number of pool tests required as well, \tifs{while providing sufficient randomness} to obtain the secrecy constraint as we prove in \Cref{LowerBoundLeakage}. This normalization is possible since in the output sequence, Eve actually \tifs{can obtain} only a sum of $K$ rows, and does not have access to any specific row.

\begin{proof}[Proof of \Cref{direct lemma1}]
\tifs{The proof is similar to applications of the Gallager-type bound \cite{gallager1968information} and specifically that in \cite[Proof of Theorem III.1]{atia2012boolean}. However, due the different code construction, the details are different. Specially, for each item there is a \emph{bin} of $M$ codewords, from which the decoder has to choose.}

\tifs{Define
\begin{equation}\label{eq:exp_bound}
\emph{f}(\rho) = E_{o}(\rho)-\rho\frac{\log\binom{N-K}{i}M^i}{T}-\frac{\log\binom{K}{i}}{T}-\frac{K}{T}.
\end{equation}
Since $0 \leq \rho \leq 1$ can be optimized, we wish to show that  $T\emph{f}(\rho)\rightarrow \infty$ as $N \rightarrow \infty$ for some $\rho$ in this range. If this is done for all $i$, then due to the resulting exponential decay of $P(E_i)$, using a simple union bound will show that the error probability is small in total as well, completing the proof of \Cref{direct lemma1}.}

\tifs{Since the function $\emph{f}(\rho)$ is differentiable and has a power series expansion, using a Taylor series expansion in the neighborhood of $\rho=0$ we have
\begin{equation*}
\emph{f}(\rho) = \emph{f}(0) + \rho f'(0) + \frac{\rho^2}{2}f''(\psi)
\end{equation*}
for some $\psi\in (0,\rho)$. Now,}
\ifdouble
\tifs{\begin{eqnarray}
&& \hspace{-0.85cm} \frac{\partial E_{o}}{\partial\rho} |_{\rho = 0}\nonumber\\
&\hspace{-0.7cm}  = &\hspace{-0.5cm} \sum_Y \sum_{X_{\mathcal{S}^2}}[\sum_{X_{\mathcal{S}^1}}P(X_{\mathcal{S}^1})p(Y,X_{\mathcal{S}^2}|X_{\mathcal{S}^1})\log p(Y,X_{\mathcal{S}^2}|X_{\mathcal{S}^1})\nonumber\\
&\hspace{-0.7cm}  - &\hspace{-0.5cm} \sum_{X_{\mathcal{S}^1}}P(X_{\mathcal{S}^1})p(Y,X_{\mathcal{S}^2}|X_{\mathcal{S}^1}) \sum_{X_{\mathcal{S}^1}}P(X_{\mathcal{S}^1})p(Y,X_{\mathcal{S}^2}|X_{\mathcal{S}^1})]\nonumber\\
&\hspace{-0.7cm}  = &\hspace{-0.5cm} \sum_Y \sum_{X_{\mathcal{S}^2}}\sum_{X_{\mathcal{S}^1}}P(X_{\mathcal{S}^1})p(Y,X_{\mathcal{S}^2}|X_{\mathcal{S}^1})\nonumber\\
&& \hspace{-0.85cm} \log\frac{p(Y,X_{\mathcal{S}^2}|X_{\mathcal{S}^1})}{\sum_{X_{\mathcal{S}^1}}P(X_{\mathcal{S}^1})p(Y,X_{\mathcal{S}^2}|X_{\mathcal{S}^1})}
%
= I(X_{\mathcal{S}^1};X_{\mathcal{S}^2},Y).\nonumber
\end{eqnarray}}
\else
\tifs{\begin{eqnarray}
&& \hspace{-0.85cm} \frac{\partial E_{o}}{\partial\rho} |_{\rho = 0}\nonumber\\
&\hspace{-0.7cm}  = &\hspace{-0.5cm} \sum_Y \sum_{X_{\mathcal{S}^2}}[\sum_{X_{\mathcal{S}^1}}P(X_{\mathcal{S}^1})p(Y,X_{\mathcal{S}^2}|X_{\mathcal{S}^1})\log p(Y,X_{\mathcal{S}^2}|X_{\mathcal{S}^1})\nonumber\\
&\hspace{-0.7cm}  - &\hspace{-0.5cm} \sum_{X_{\mathcal{S}^1}}P(X_{\mathcal{S}^1})p(Y,X_{\mathcal{S}^2}|X_{\mathcal{S}^1}) \sum_{X_{\mathcal{S}^1}}P(X_{\mathcal{S}^1})p(Y,X_{\mathcal{S}^2}|X_{\mathcal{S}^1})]\nonumber\\
&\hspace{-0.7cm}  = &\hspace{-0.5cm} \sum_Y \sum_{X_{\mathcal{S}^2}}\sum_{X_{\mathcal{S}^1}}P(X_{\mathcal{S}^1})p(Y,X_{\mathcal{S}^2}|X_{\mathcal{S}^1}) \log\frac{p(Y,X_{\mathcal{S}^2}|X_{\mathcal{S}^1})}{\sum_{X_{\mathcal{S}^1}}P(X_{\mathcal{S}^1})p(Y,X_{\mathcal{S}^2}|X_{\mathcal{S}^1})}\nonumber\\
&\hspace{-0.7cm}  = &\hspace{-0.5cm} I(X_{\mathcal{S}^1};X_{\mathcal{S}^2},Y).\nonumber
\end{eqnarray}}
\fi
\tifs{Hence, since $E_o(0)=0$, we have
\begin{multline*}
    T\emph{f}(\rho) = T\rho\Bigg(I(X_{\mathcal{S}^1};X_{\mathcal{S}^2},Y) - \frac{\log\binom{N-K}{i}M^i}{T}\Bigg)\\ - \log\binom{K}{i}-K+T\frac{\rho^2}{2}E_o''(\psi).
\end{multline*}
Consequently, if $T  \geq (1+\varepsilon)\max_{i=1,\ldots ,K} \frac{\log\binom{N-K}{i}M^i}{I(X_{\mathcal{S}^1};X_{\mathcal{S}^2},Y)}$ for some $\varepsilon >0$, we have
\begin{multline*}
    T\emph{f}(\rho) \ge \\ T\rho\Bigg(I(X_{\mathcal{S}^1};X_{\mathcal{S}^2},Y) \left(\frac{\varepsilon}{1+\varepsilon}\right) + \frac{\rho}{2}E_o''(\psi) \Bigg)  -\log\binom{K}{i}-K.
\end{multline*}
Note that $E_o''(\psi)$ is negative \cite{gallager1968information}. However, it is independent of the other constants and $T$, hence choosing 
\[
0 < \rho < \frac{2I(X_{\mathcal{S}^1};X_{\mathcal{S}^2},Y)\left(\frac{\varepsilon}{1+\varepsilon}\right)}{|E_o''(\psi)|}
\]
and remembering that both $\log\binom{K}{i}$ and $K$ are fixed and independent of $T$ as well, we have $T\emph{f}(\rho) \rightarrow \infty$ as $N$ (and consequently $T$) go to $\infty$.}
\end{proof}
The expression $I(X_{\mathcal{S}^1};X_{\mathcal{S}^2},Y)$ in \Cref{direct lemma1} is critical to understand how many tests are required, yet it is not a function of the problem parameters in any straight forward \tifs{manner}. We now bound it to get a better handle on $T$.
\begin{claim}\label{InformationClaim}
For large $K$, and under a fixed input distribution for the testing matrix $(\frac{\ln(2)}{K},1-\frac{\ln(2)}{K})$, the mutual information between $X_{\mathcal{S}^1}$ and $(X_{\mathcal{S}^2},Y)$ is lower bounded by
\begin{equation*}
I(X_{\mathcal{S}^1};X_{\mathcal{S}^2},Y)\geq \frac{i}{K}.
\end{equation*}
\end{claim}
\begin{proof}[Proof of \Cref{InformationClaim}] First, note that
\begin{eqnarray*}\label{eq:MutualInfB}
I(X_{\mathcal{S}^1};X_{\mathcal{S}^2},Y)
&&\hspace{-0.6cm} = I(X_{\mathcal{S}^1};X_{\mathcal{S}^2})+I(X_{\mathcal{S}^1};Y|X_{\mathcal{S}^2})\nonumber\\
&&\hspace{-0.6cm} \stackrel{(a)}{=} H(Y|X_{\mathcal{S}^2})-H(Y|X_{\mathcal{S}})\nonumber\\
&&\hspace{-0.6cm} \stackrel{(b)}{=} q^{K-i}H(q^{i})\nonumber\\
&&\hspace{-0.6cm} = q^{K-i}\left[q^{i}\log\frac{1}{q^{i}}
+  \left( 1 - q^{i}\right)\log\frac{1}{\left(1- q^{i}\right)}\right],
\end{eqnarray*}
where equality (a) follows since the rows of the testing matrix are independent, and (b) follows since $H(Y|X_{\mathcal{S}})$ is the uncertainty of the legitimate receiver given $X_{\mathcal{S}}$, thus when observing the noiseless outcomes of all pool tests, this uncertainty is zero. Also, note that the testing matrix is random and i.i.d.\ with distribution $(1-q,q)$, hence the probability for $i$ zeros is $q^i$.

Then, under a fixed input distribution for the testing matrix $(p=\frac{\ln(2)}{K},q=1-\frac{\ln(2)}{K})$ and large $K$ it is easy to verify that the bounds meet at the two endpoint of $i=1$ and $i=K$ (since $H(q) \to 0$ and $H(q^K) \to 1$ as $K$ grows), yet the mutual information is concave in $i$ (extending $i$ from integers to reals in $[1,K]$ to make this notion meaningful) thus the bound is obtained. This is demonstrated graphically in \Cref{figure:Mutual_Information_Bound}.
\end{proof}

Applying \Cref{InformationClaim} to the expression in \Cref{direct lemma1}, we have
\begin{equation}\label{eq:bound_apply}
\frac{\log\tifs{\left(\binom{N-K}{i}M^i\right)}}{I(X_{\mathcal{S}^1};X_{\mathcal{S}^2},Y)} \leq \frac{\log\tifs{\left(\binom{N-K}{i}M^i\right)}}{\frac{i}{K}}.
\end{equation}
Hence, substituting $M=2^{T\frac{\delta-\epsilon_K}{K}}$, a sufficient condition for reliability is
\begin{eqnarray*}
T &\ge& \max_{1 \leq i \leq K}\frac{1+\varepsilon}{\frac{i}{K}}\left[ \log\binom{N-K}{i} + \frac{i}{K}T(\delta-\epsilon_K) \right]
\end{eqnarray*}
Rearranging terms results in
\begin{eqnarray}\label{eq:reduce_h}
T & \ge & \max_{1 \leq i \leq K} \frac{1}{1-\delta +\epsilon_K-\varepsilon\delta+\varepsilon\epsilon_K}\frac{1+\varepsilon}{i/K}\log\binom{N-K}{i},\nonumber
\end{eqnarray}
where by reducing $\epsilon_K$ and $\varepsilon\epsilon_K$ we increase the bound on $T$, and with some constant $\varepsilon > 0$.
Noting that this is for large $K$ and $N$, and that $\varepsilon$ is independent of them. Hence, by replacing $\log\binom{N-K}{i}$ by it's asymptotic value $i\log N$ we achieves the bound on $T$ provided in \Cref{direct theorem1},
\[
 \frac{1 + \varepsilon}{1 - \delta} K \log N,
\]
and reliability is established.

\ifdouble
\begin{figure}
  \centering
  \includegraphics[trim=2.2cm 14.4cm 1.5cm 7.2cm,clip,scale=0.52]{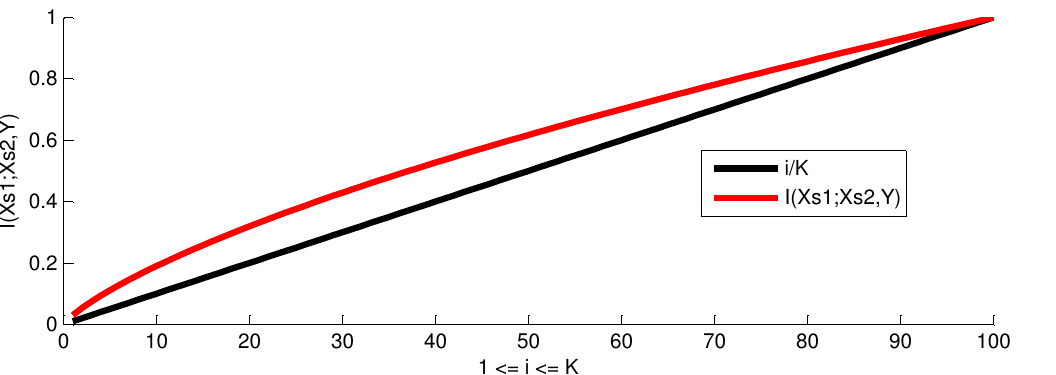}
  \caption{Mutual Information Bound \tifs{for $K=100$, and under a fixed input distribution for the testing matrix $(\frac{\ln(2)}{K},1-\frac{\ln(2)}{K})$}.}
  \label{figure:Mutual_Information_Bound}
  \vspace{-0.4cm}
\end{figure}
\else
\begin{figure}
  \centering
  \includegraphics[trim=2.2cm 14.4cm 1.5cm 7.2cm,clip,scale=0.8]{MATLAB.pdf}
  \caption{Mutual Information Bound}
  \label{figure:Mutual_Information_Bound}
\end{figure}
\fi

\subsection{Information Leakage at the Eavesdropper}\label{LowerBoundLeakage}
We now prove the security constraint is met. Hence, we wish to show that $\frac{1}{T}I(W;Z^{T})\rightarrow 0$, as $T\rightarrow \infty$. \tifs{This is done in two major steps.} Denote by $\mathcal{C}_T$ the random codebook and by $X_{\mathcal{S}}^T$ the set of codewords corresponding to the true, defective items set $\mathcal{S}$, given $W$ and the corresponding function $(W;R_K) \rightarrow X^{T}_{\mathcal{S}}$. \tifs{In the first step, we bound $\frac{1}{T}I(W;Z^{T})$ by $\frac{1}{T}H(R_K|Z^T,W,\mathcal{C}_{T})$, using a technical chain of inequalities. In the second step, however, we wish to show that this resulting expression actually depicts a key property of physical layer security: given the codebook, the message and Eve's observations, Eve is supposed to be able to estimate the randomness used, $R_K$. Hence, this conditional entropy is low, and as a result - the mutual information.}
\tifs{\subsubsection*{Step 1} We have,}
\begin{eqnarray}
&& \hspace{-0.8cm}\textstyle\textstyle \frac{1}{T}I(W;Z^T|\mathcal{C}_{T}) = \frac{1}{T}\left(I(W,R_K;Z^T|\mathcal{C}_{T})-I(R_K;Z^T|W,\mathcal{C}_{T})\right)\nonumber\\
&\hspace{-0.6cm}  \stackrel{(a)}{\leq} &\hspace{-0.5cm}\textstyle \frac{1}{T}\left(I(X_{\mathcal{S}}^T;Z^T|\mathcal{C}_{T})-I(R_K;Z^T|W,\mathcal{C}_{T})\right)\nonumber\\
&\hspace{-0.6cm} = &\hspace{-0.5cm}\textstyle \frac{1}{T}(I(X_{\mathcal{S}}^T;Z^T|\mathcal{C}_{T})-H(R_K|W,\mathcal{C}_{T})+H(R_K|Z^T,W,\mathcal{C}_{T}))\nonumber\\
&\hspace{-0.6cm} \stackrel{(b)}{=} & \hspace{-0.5cm}\textstyle \frac{1}{T}\left(I(X_{\mathcal{S}}^T;Z^T|\mathcal{C}_{T})-H(R_K)+H(R_K|Z^T,W,\mathcal{C}_{T})\right)\nonumber\\
&\hspace{-0.6cm} \stackrel{(c)}{\leq}&\hspace{-0.5cm} I(X_{\mathcal{S}};Z|\mathcal{C}_{T}) - \frac{1}{T}K\log M + \frac{1}{T}H(R_K|Z^T,W,\mathcal{C}_{T})\label{eq:leak_apply}\hspace{-0.2cm}\\
&\hspace{-0.6cm} \stackrel{(d)}{\leq}&\hspace{-0.5cm} \delta - \frac{1}{T}K\left(T \frac{\delta-\epsilon^{\prime}}{K}\right) + \frac{1}{T}H(R_K|Z^T,W,\mathcal{C}_{T})
\nonumber\\
&\hspace{-0.6cm} =&\hspace{-0.5cm} \epsilon^{\prime}+ \frac{1}{T}H(R_K|Z^T,W,\mathcal{C}_{T}).\nonumber
\end{eqnarray}
(a) is since $(W;R_K) \rightarrow X^{T}_{\mathcal{S}} \rightarrow Z^T$; note that given $\mathcal{C}_{T}$ there is a function correspondence between $(W,R_K)$ and $X_{\mathcal{S}}^T$, the mapping of the defective set and the internal randomness to the codewords used. It is important to note that since the codebook is generated randomly, the mapping is not exactly $1:1$ and there is a possibility for a repetition of codewords. However, averaged over all the possible sets of defective items $W$, the error probability from such a repetition is negligible, and this is a direct consequence using standard analysis of random coding given in \cite[Section 3.4]{C13}; (b) is since $R_K$ is independent of $W$ and $\mathcal{C}_T$; (c) is since the channel is memoryless, hence $I(X_{\mathcal{S}}^T;Z^T|\mathcal{C}_{T}) \leq T I(X_{\mathcal{S}};Z|\mathcal{C}_{T})$. This is a standard application of, e.g., \cite[Lemma 7.9.2]{cover2012elements}; (d) is since by choosing an i.i.d.\ distribution for the random codebook, when $H(R_K)=\log(M^K)$ and $R_K$ is chosen random and independent at the mixer, one easily observes that $I(X_{\mathcal{S}};Z| \mathcal{C}_T) \leq \delta$.
\tifs{\subsubsection*{Step 2} To bound $H(R_K|Z^T,W,\mathcal{C}_T)$, we wish to show that} \emph{given the set of defective items, the codebook, and the information Eve has}, she can infer the internal randomess used with a low error probability. This is, in a sense, similar to wiretap coding or secure network coding, where the message itself is concealed from Eve, using the internal randomness, \emph{yet given the message, Eve can easily guess the internal randomness used}. Now, for the specific case herein: The codebook is known to all. Any party knowing $W$, knows which were the defective items (from which bins the codewords were taken). We claim that this situation is \emph{analogous to a noiseless Boolean Multiple Access Channel, followed by an erasure channel}. To see this, note that there are $K$ defectives. Each of these (a ``user", in the analogy) puts a codeword on the channel. The codewords are summed. Then, Eve sees the sum through an erasure channel, since she sees only a fraction $\delta$. We wish to show that Eve can decode, i.e., know which codeword each defective item has put on the channel, hence get the internal randomness. The channel capacity is $\delta$. There are $K$ users. Thus, a rate of $\delta/K$ is attainable for each user. Since the block size is $T$, each user can have (a bit less than) $2^{T(\frac{\delta}{K})}$ codewords and Eve can still infer which codeword each user transmitted. \tifs{In other words, if $R_K$ can be estimated from $(Z^T,W,\mathcal{C}_T)$ with an error probability that goes to $0$ as $T \to \infty$, we have $H(R_K|Z^T,W,\mathcal{C}_T) \leq T\epsilon_T$ for some $\epsilon_T$ which satisfies $\epsilon_T \to 0$ as $T \to \infty$. This is a standard application of Fano's inequality \cite[Section 2.10]{cover2012elements}, with a specific Multiple Access Channel (MAC) application in \cite[Section 15.3.4]{cover2012elements}. Thus, $H(R_K|Z^T,W,\mathcal{C}_T) \leq T \epsilon_T$, where $\epsilon_T \to 0$ as $T \to \infty$}. Finally, all that is left to explain is that the codewords used for SGT are also good to attain the capacity for this MAC. The distribution used herein is i.i.d. with probabilities  $(\frac{\ln(2)}{K},1-\frac{\ln(2)}{K})$.\tifs{\footref{note10}} The MAC requires i.i.d. with probabilities $(1-2^{\frac{1}{K}}, 2^{\frac{1}{K}})$, but since the mutual information is continuous in the input distribution, this holds.

\begin{remark}
Under non-secure GT, it is clear that simply adding tests to a given GT code (increasing $T$) can only improve the performance of the code (in terms of reliability). A legitimate decoder can always disregard the added tests. For SGT, however, the situation is different. Simply adding tests to a given code, \emph{while fixing the bin sizes}, might make the vector of results vulnerable to eavesdropping. In order to increase reliability, one should, of course, increase $T$, but also increase the bin sizes proportionally, so the secrecy result above will still hold. This will be true for the efficient algorithm suggested in \Cref{efficient_algorithms} as well.
\end{remark}

\begin{remark}
Note that since
\begin{eqnarray*}
  \frac{1}{T}H(R_K|Z^T,W,\mathcal{C}_{T}) \hspace{-0.2cm} &=\hspace{-0.2cm}& \frac{1}{T}H(R_K)-\frac{1}{T}I(R_K;Z^T,W,\mathcal{C}_{T}) \\
   \hspace{-0.2cm}&=\hspace{-0.2cm}& \frac{1}{T}\tifs{K}\log M - \frac{1}{T}I(R_K;Z^T,W,\mathcal{C}_{T}),
\end{eqnarray*}
any finite-length approximation for $I(R_K;Z^T,W,\mathcal{C}_{T})$ will give a finite length approximation for the leakage at the eavesdropper.
For example, one can use the results in \cite{polyanskiy2010channel}, to show that the leakage can be approximated as $\frac{1}{\sqrt{T}}+\epsilon^{\prime}$.
\end{remark}
\section{Converse (Necessity)}\label{converse}
In this section, we derive the necessity bound on the required number of tests.
Let $\bar{Z}^T$ denote the random variable corresponding to the tests which are not available to the eavesdropper. Hence, $Y^T = (Z^T,\bar{Z}^T)$.
Similar to the non-secure GT setting, e.g., \cite[Section IV]{atia2012boolean} and \cite[Section IV]{chan2011non}, by Fano's inequality, if $P_e\rightarrow 0$ when $T\rightarrow\infty$, then
\begin{equation*}
H(W|Y^T) \leq  \log{\binom{N}{K}}\epsilon'_T,
\end{equation*}
for some $\epsilon'_T \rightarrow 0$ as $T\rightarrow\infty$. Moreover, the \tifs{weak secrecy constraint\footnote{\tifs{\Cref{whyW} elaborates why weak secrecy is considered in this section.}}, as proved in \Cref{LowerBoundLeakage}}, implies
\begin{equation}\label{eq:Necessity19}
I(W;Z^T)\leq T\epsilon''_T,
\end{equation}
where $\epsilon''_T\rightarrow 0$  as $T\rightarrow \infty$.  

Consequently,
\begin{eqnarray}\label{eq:R_s}
\log \binom{N}{K} &=& H(W)
\nonumber\\
&=& I(W;Y^T) + H(W|Y^T)
\nonumber\\
&\stackrel{(a)}{\leq}& I(W;Z^T,\bar{Z}^T) + \log{\binom{N}{K}}\epsilon'_T
\nonumber\\
&=& I(W;Z^T) + I(W;\bar{Z}^T|Z^T) + \log{\binom{N}{K}}\epsilon'_T
\nonumber\\
&\stackrel{(b)}{\leq} & I(W;\bar{Z}^T|Z^T) + \log{\binom{N}{K}}\epsilon'_T + T\epsilon''_T 
\nonumber\\
&=& H(\bar{Z}^T|Z^T) -
\nonumber\\
&&\qquad H(\bar{Z}^T|W,Z^T) + \log{\binom{N}{K}}\epsilon'_T + T\epsilon''_T
\nonumber\\
&\stackrel{(c)}{\leq} & H(\bar{Z}^T) + \log{\binom{N}{K}}\epsilon'_T + T\epsilon''_T
\nonumber
\end{eqnarray}
where (a) follows from Fano's inequality and since $Y^T = (Z^T,\bar{Z}^T)$, (b) follows from \eqref{eq:Necessity19}\tifs{, and (c) follows because conditioning reduces entropy. Note that the use of Fano in (a) is weaker than the bound in \cite[Theorem 3.1]{baldassini2013capacity}, which states that for non-secure GT, $P(\text{success}) \leq 2^T/\binom{N}{K}$. However, both predict the same phase transition from a high error probability to a low one, hence in the context of this converse both suffice.}

We now evaluate $H(\bar{Z}^T)$. Denote by $\mathcal{\bar{E}}$ the set of tests which are not available to Eve and by $\bar{E}_\gamma$ the event $\{|\mathcal{\bar{E}}| \leq T(1-\delta)(1+\gamma)\}$ for some $\gamma >0$. Let $\textbf{1}_{\bar{E}_\gamma}$ be the indicator for this event. We have
\begin{eqnarray*}
    H(\bar{Z}^T) &\leq&H(\bar{Z}^T, \textbf{1}_{\bar{E}_\gamma})
\\
    &\leq&H(\bar{Z}^T| \textbf{1}_{\bar{E}_\gamma})+1
\\
	&\stackrel{(a)}{=}& P(\bar{E}_\gamma)H(\bar{Z}^T| \bar{E}_\gamma) + P(\bar{E}^c_\gamma)H(\bar{Z}^T| \bar{E}^c_\gamma)+1
\\
	&\stackrel{(b)}{\leq}& T(1-\delta)(1+\gamma) + TP(\bar{E}^c_\gamma)+1
	\\
	& \stackrel{(c)}{\leq}& T(1-\delta)(1+\gamma) + T2^{-T(1-\delta)f(\gamma)}+1,
\end{eqnarray*}
where (a) follows from the law of total expectation, applied to the entropy of $\bar{Z}^T$ given the indicator. (b) follows since in the first summand, $P(\bar{E}_\gamma)$ can surely be bounded by 1, while given $\bar{E}_\gamma$, the number of tests which are not available to Eve is at most $T(1-\delta)(1+\gamma)$, hence their entropy is at most that number (as each test result is binary); for the second summand, the entropy of $\bar{Z}^T$ (even without conditioning) is clearly bounded by $T$, as it is a binary vector of length at most $T$. (c) follows from the Chernoff bound for i.i.d.\ Bernoulli random variables with parameter $(1-\delta)$ and is true for some $f(\gamma)$ such that $f(\gamma)>0$ for any $\gamma > 0$.

Thus, we have
\ifdouble
\begin{multline*}
    \log \binom{N}{K} \leq T(1-\delta)(1+\gamma) \\+ T2^{-T(1-\delta)f(\gamma)} +
\log{\binom{N}{K}}\epsilon'_T + T\epsilon''_T+1.
\end{multline*}
\else
\[
\log \binom{N}{K} \leq T(1-\delta)(1+\gamma) + T2^{-T(1-\delta)f(\gamma)} +\log{\binom{N}{K}}\epsilon'_T + T\epsilon''_T+1.
\]
\fi
That is,
\[
	T \ge \frac{1-\epsilon_T}{1-\delta}\log \binom{N}{K},
	\]
for some $\epsilon_T$ such that $\epsilon_T \to 0$ as $T \to \infty$. This completes the converse proof.
 %
\section{Noisy observation at the eavesdropper}\label{noisy_eve}
To simplify the technical aspects and allow us to focus on the key methods and results, during this paper we consider an noiseless channel at the legitimate receiver and erasure channel at the eavesdropper.
Yet, the outcome vector $Y^T$ and the noisy vector $Z^T$, may be generated from a channels with other noise models.

The case of noisy vector $Y^T$ at the decoder in the non-secure problem was considered in the literature for many models of noises (e.g., \cite{atia2012boolean}). Such that, the bound on $T$ in the noisy case is given by bounding the mutual information between the mixer and the decoder. We assume that using similar technics as in the non-secure problem to bound the rates of the mutual information at the legitimate decoder in the secure problem we consider herein the bound on $T$ is simply given as well.

In this section, \tifs{under the weak security guarantee in \Cref{LowerBoundLeakage}}, we will generalize the result given in \Cref{LowerBound} to address the case where the $Z^T$ is generated from any possible model of noise, e.g., false positive errors, false negative errors, both possible errors or a Binary Symmetric Channel (BSC), as considered in \cite{atia2012boolean}, \cite{sejdinovic2010note} and \cite{chan2011non}, respectively. The information obtained at the eavesdropper from the noisy observation $Z$, is $I(X_{\mathcal{S}};Z)$. The codebook, testing process and the decoding algorithm at the legitimate decoder are as given \Cref{LowerBound}, yet, we chose $M$, the number of rows in each bin, such that $\log_2(M) = T(I(X_{\mathcal{S}};Z)-\epsilon_K)/K$. 

Substituting $M=2^{T\frac{I(X_{\mathcal{S}};Z)-\epsilon_K}{K}}$ to the inequality given in \eqref{eq:bound_apply}, a sufficient condition for reliability is
\begin{eqnarray*}
T &\ge& \max_{1 \leq i \leq K}\frac{1+\varepsilon}{\frac{i}{K}}\left[ \log\binom{N-K}{i} + \frac{i}{K}T(I(X_{\mathcal{S}};Z)-\epsilon_K) \right].
\end{eqnarray*}
Note that this step is valid since the observations at the legitimate side are noiseless. Rearranging terms results in
\begin{eqnarray}
T & \ge & \max_{1 \leq i \leq K} \frac{(1+\varepsilon)\log\binom{N-K}{i}}{(1-I(X_{\mathcal{S}};Z) +\epsilon_K-\varepsilon I(X_{\mathcal{S}};Z)+\varepsilon\epsilon_K)i/K},\nonumber
\end{eqnarray}
where by reducing $\epsilon_K$ and $\varepsilon\epsilon_K$ we increase the bound on $T$, and with some constant $\varepsilon > 0$.
Noting that this is for large $K$ and $N$, and that $\varepsilon$ is independent of them, achieves the bound on $T$ for any model of noise given at the outcome of the eavesdropper and reliability is established.

We now briefly show how the weak security guarantee, proved in \tifs{\Cref{LowerBoundLeakage}} for the noiseless setting, can be applied here as well. In fact, substituting $M=2^{T\frac{I(X_{\mathcal{S}};Z)-\epsilon_K}{K}}$  to \eqref{eq:leak_apply},
\ifdouble
\begin{multline}
  I(X_{\mathcal{S}};Z|\mathcal{C}_{T}) - \frac{1}{T}K\left(T \frac{I(X_{\mathcal{S}};Z)-\epsilon_K}{K}\right)\\ + \frac{1}{T}H(R_K|Z^T,W,\mathcal{C}_{T}) \stackrel{(a)}{\leq}  \epsilon_T + \epsilon_K,
\end{multline}
\else
\begin{equation}
I(X_{\mathcal{S}};Z|\mathcal{C}_{T}) - \frac{1}{T}K\left(T \frac{I(X_{\mathcal{S}};Z)-\epsilon_K}{K}\right) + \frac{1}{T}H(R_K|Z^T,W,\mathcal{C}_{T}) \stackrel{(a)}{\leq}  \epsilon_T + \epsilon_K,
\end{equation}
\fi
where (a) is given from the proof in \tifs{\Cref{LowerBoundLeakage}}.

For example, to illustrate the case considered in this section we will analyze the possible false positive errors noise at the eavesdropper. Such that, where $U(t)\sim Bernoulli(u)$, the eavesdropper observes
\begin{equation*}
Z(t)=Y(t)\vee U(t)= \left(\bigvee_{j=1}^{N}X_{j}(t)A_j\right)\vee U(t),\text{ }\text{ } t=1,\ldots,T.
\end{equation*}
For possible false positive errors, large $K$, and under a fixed input distribution for the testing matrix $(\frac{\ln(2)}{K},1-\frac{\ln(2)}{K})$, the mutual information between $X_{\mathcal{S}}$ and $Z$ is bounded by
\begin{eqnarray*}
  I(X_{\mathcal{S}};Z) &=& H(Z) - H(Z|X_{\mathcal{S}})\\
  & = & H((1-p)^{K}(1-u))- (1-p)^{K} H(u)\\
  & = & H(0.5(1-u))- 0.5 H(u).
\end{eqnarray*}
Substituting this mutual information to the bound on $T$ given in this section for any noise we have,
\begin{equation*}
 T  \geq  \max_{i=1,\ldots ,K} \frac{(1+\varepsilon)\log\binom{N-K}{i}}{1-(H(0.5(1-u))- 0.5 H(u))} \frac{K}{i}.
\end{equation*}
By replacing $\log\binom{N-K}{i}$ by it's asymptotic value $i\log N$,
\[
  T  \geq  \frac{(1+\varepsilon)}{1-(H(0.5(1-u))- 0.5 H(u))} K \log N.
\]
\\That is, for some $\varepsilon > 0$ and an eavesdropper with false positive possible errors, as $N\rightarrow \infty$, there exists a sequence of SGT algorithms which are reliable and secure.

To conclude, it is not hard to see that by setting $M$ according to the mutual information between $X_{\mathcal{S}}$ to $Z$ in any model of noise at the eavesdropper we can bound the number of tests to get both, reliability and secrecy constraints.      %
\section{Efficient Algorithms}\label{efficient_algorithms}
The achievability result given in \Cref{direct theorem1} uses a random codebook and ML decoding at the legitimate party. The complexity burden in ML, however, prohibits the use of this result for large $N$. In this section, we derive and analyze an efficient decoding algorithm, which maintains the reliability result using a much simpler decoding rule, at the price of \emph{only slightly more tests}. 

\tifs{The codebook construction and the mixing process do not change compared to the achievability result given before. Particularly, note that the information leakage proof given in \Cref{LowerBoundLeakage} is independent of the decoding algorithm, and valid for all $T$. Essentially, with the correct value of $M$ (which, of course, depends on the other parameters of the problem), whatever Eve does, her normalized mutual information is small. Thus, as long as the construction is kept and $M$ is fixed the same way, the security proof holds.} Moreover, the result in this section will hold for any $K$, including even the case were $K$ grows linearly with $N$.

Specifically, we assume the same codebook generation and the testing procedure given in \Cref{LowerBound}, \tifs{and analyze the \emph{Definitely Non-Defective} (DND) algorithm, previously considered in the literature for the non-secure case \cite{kautz1964nonrandom,chen2008survey,chan2011non,chan2014non,aldridge2014group,lo2013efficient,malyutov2013search,aldridge2019group}.}
The decoding algorithm at the legitimate user is as follows. Bob attempts to match the rows of \textbf{X} with the outcome vector $Y^T$. If a particular row $j$ of \textbf{X} has the property that all locations $t$ where it has $1$, also corresponds to a $1$ in $Y(t)$, then that row \emph{can correspond to a defective item}. If, however, the row has $1$ at a location $t$ where the output has $0$, then it is not possible that the row corresponds to a defective item. The problem, however, when considering the code construction in this paper for SGT, is that the decoder does not know which row from each bin was selected for any given item. Thus, it takes a conservative approach, and declares an item as defective if at least one of the rows in its bin signals it may be so. An item is not defective only if \emph{all the rows in its bin prevent it from being so}.

It is clear that this decoding procedure has no false negatives, as a defective item will always be detected. It may have, though, false positives. A false positive may occur if all locations with ones in a row corresponding to a non-defective item are hidden by the ones of other rows corresponding to defective items and selected by the mixer. To calculate the error probability, fix a row of $\textbf{X}$ corresponding to a non-defective item (a row in its bin). Let $j_1;\ldots;j_k$ index the rows of $\textbf{X}$ corresponding to the $K$ defective items, and selected by the mixer for these items (that is, the rows which were actually added by the Boolean channel).
An error event associated with the fixed row occurs if at any test where that row has a $1$, at least one of the entries $X_{j_1}(t), \ldots, X_{j_k}(t)$ also has a $1$. The probability for this to happen, per column, is $p(1-(1-p)^K)$. Hence, the probability that a test result in a fixed row is hidden from the decoder, in the sense that it cannot be declared as non defective due to a specific column, is
\[
p(1-(1-p)^K)+(1-p) = 1-p(1-p)^K.
\]
Since this should happen for all $T$ columns, the error probability for a fixed row is $\left(1-p(1-p)^K\right)^T$. This error probability, for a fixed row, is similar to the non-secure case given in  \cite{chan2014non,aldridge2014group}. However, to compute the error probability for the entire procedure in the secure setting, we must take a union bound over all $M(N-K)$ rows in the codebook, unlike a union bound on the $N-K$ unique rows corresponding to non-defective items. Moreover, to provide bounds on the error probability and hence $T$, in the secure setting, we utilize a specific input distribution as given in \Cref{LowerBound} and the number of codewords $M$ in each bin. This, according to the following analysis for SGT, leads to an additional factor of $1/2-\delta$, which is critical and shows a dependence on Eve's fraction of leaked test results. Thus, we have
\ifdouble
\begin{figure}
  \centering
  \includegraphics[trim= 1.6cm 0.0cm 0cm 0.5cm,clip,scale=0.228]{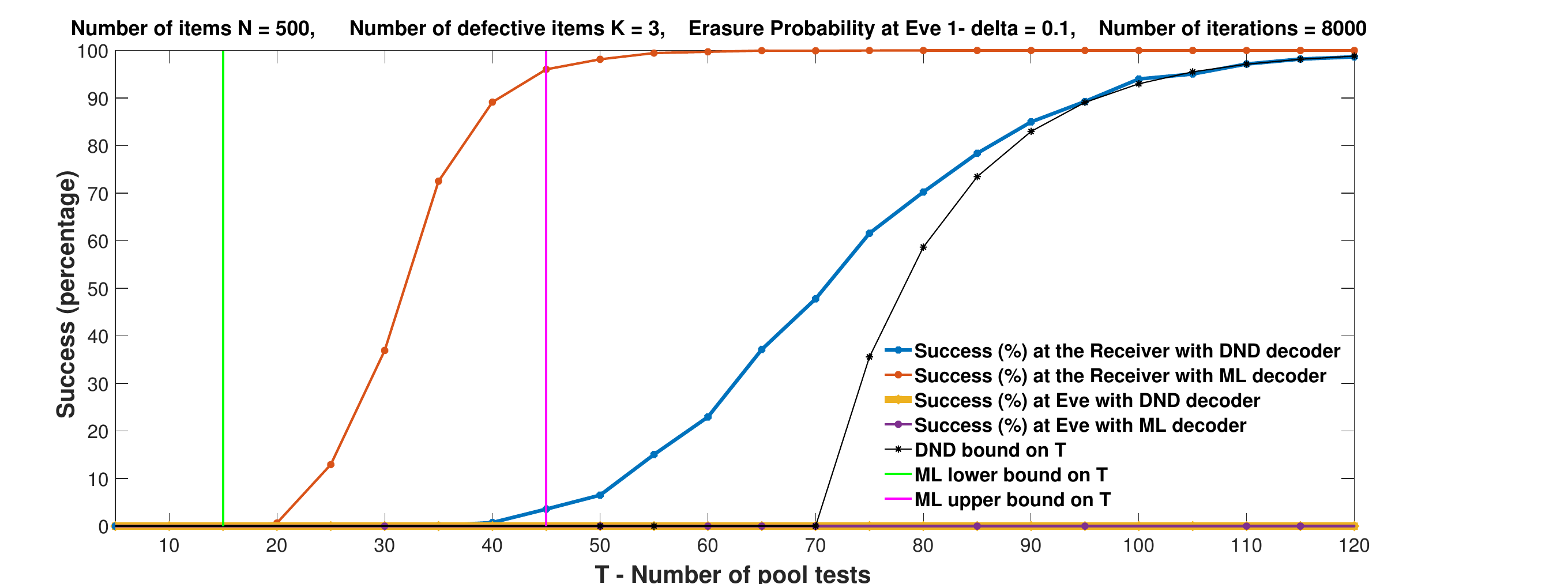}
  \caption{\emph{Definitely Non-Defective} and ML simulation results.}
  \label{fig:DND_simulation}
  \vspace{-0.4cm}
\end{figure}
\else
\begin{figure}
  \centering
  \includegraphics[trim= 0cm 0.0cm 0cm 0.5cm,clip,scale=0.4]{Success_Prob_of_all_the_DND_ML_items_N_500_K_3_Delta_01_NumIterations_8000_date_2017_3_11_6_54_v3.eps}
  \caption{\emph{Definite Non-Defective} and ML simulation results.}
  \label{fig:DND_simulation}
\end{figure}
\fi
\ifdouble
\begin{eqnarray}\label{eq:PeComp}
&& \hspace{-0.9cm} P_e \leq M(N-K)\left(1-p(1-p)^{K}\right)^T
\nonumber\\
&& \hspace{-0.5cm} \stackrel{(a)}{\leq} MN\left(1-\frac{\tifs{\ln2}}{K}\left(1-\frac{\tifs{\ln2}}{K}\right)^{K}\right)^{\beta K\log N}
\nonumber\\
&&\hspace{-0.5cm} = MN\left(1-\frac{\tifs{\ln2}}{K}\left(1-\frac{\tifs{\ln2}}{K}\right)^{K-1}\left(1-\frac{\tifs{\ln2}}{K}\right)\right)^{\beta K\log N}
\nonumber\\
&& \hspace{-0.5cm} \stackrel{(b)}{\leq} MN\left(\left(1-\frac{\tifs{\ln2}\left(1-\frac{\tifs{\ln2}}{K}\right)}{Ke^{\tifs{\ln2}}}\right)^{K}\right)^{\beta\log N}
\nonumber\\
&& \hspace{-0.5cm} \stackrel{(c)}{\leq} MNe^{-\tifs{\ln2}\left(1-\frac{\tifs{\ln2}}{K}\right)e^{-\tifs{\ln2}}\beta\log N}\nonumber\\
&&\hspace{-0.5cm} \leq MN^{1-\tifs{\ln2}\left(1-\frac{\tifs{\ln2}}{K}\right)e^{-\tifs{\ln2}}\beta\frac{1}{\ln 2}}
\nonumber\\
&& \hspace{-0.5cm} = MN^{1-\frac{1}{2}\beta (1-\frac{\ln2}{K})}
\nonumber\\
&&\hspace{-0.5cm} \stackrel{(d)}{=} 2^{\beta (\delta-\epsilon) \log N}N^{1-\frac{1}{2}\beta (1-\frac{\ln2}{K})}
\nonumber\\
&&\hspace{-0.5cm} = N^{\beta (\delta-\epsilon)}N^{1-\frac{1}{2}\beta (1-\frac{\ln2}{K})}
\nonumber\\
&&\hspace{-0.5cm} \leq N^{1-\beta\left(\frac{1}{2}(1-\frac{\ln 2}{K})-\delta \right)}.
\end{eqnarray}
\else
\begin{eqnarray}\label{eq:PeComp}
  P_e &\leq& M(N-K)\left(1-p(1-p)^{K}\right)^T
\nonumber\\
& \stackrel{(a)}{\leq}& MN\left(1-\frac{\tifs{\ln2}}{K}\left(1-\frac{\tifs{\ln2}}{K}\right)^{K}\right)^{\beta K\log N}
\nonumber\\
&=& MN\left(1-\frac{\tifs{\ln2}}{K}\left(1-\frac{\tifs{\ln2}}{K}\right)^{K-1}\left(1-\frac{\tifs{\ln2}}{K}\right)\right)^{\beta K\log N}
\nonumber\\
& \stackrel{(b)}{\leq}& MN\left(\left(1-\frac{\tifs{\ln2}\left(1-\frac{\tifs{\ln2}}{K}\right)}{Ke^{\tifs{\ln2}}}\right)^{K}\right)^{\beta\log N}
\nonumber\\
& \stackrel{(c)}{\leq}& MNe^{-\tifs{\ln2}\left(1-\frac{\tifs{\ln2}}{K}\right)e^{-\tifs{\ln2}}\beta\log N}\nonumber\\
&\leq& MN^{1-\tifs{\ln2}\left(1-\frac{\tifs{\ln2}}{K}\right)e^{-\tifs{\ln2}}\beta\frac{1}{\ln 2}}
\nonumber\\
& = & MN^{1-\frac{1}{2}\beta (1-\frac{\ln2}{K})}
\nonumber
\end{eqnarray}
\begin{eqnarray}
&\hspace{-3.4cm} \stackrel{(d)}{=}& 2^{\beta (\delta-\epsilon) \log N}N^{1-\frac{1}{2}\beta (1-\frac{\ln2}{K})}
\nonumber\\
&\hspace{-3.4cm}=&N^{\beta (\delta-\epsilon)}N^{1-\frac{1}{2}\beta (1-\frac{\ln2}{K})}
\nonumber\\
&\hspace{-3.4cm}\leq&N^{1-\beta\left(\frac{1}{2}(1-\frac{\ln 2}{K})-\delta \right)}.
\end{eqnarray}
\fi
In the above, (a) follows by taking $p=\tifs{\ln2}/K$ and setting $T$ as $\beta K\log N$, for some positive $\beta$, to be defined. (b) follows since $e^{-\tifs{\ln2}} \leq (1-\tifs{\ln2}/n)^{n-1}$ for any integer $n >0$. (c) follows since $e^{-x} \ge (1-x/n)^{n}$ for $x>0$ and any integer $n >0$. (d) is by setting $M=2^{T\frac{\delta-\epsilon}{K}}$ and substituting the value for $T$.

The result in \eqref{eq:PeComp} can be interpreted as follows. As long as $\delta$, the leakage probability at the eavesdropper, is smaller than $\frac{1}{2}(1-\frac{\ln 2}{K})$, choosing $T=\beta K \log N$ with a large enough $\beta$ results in an exponentially small error probability. For example, for large enough $K$ and $\delta=0.25$, one needs $\beta > 4$, that is, about $4K\log N$ tests to have an exponentially small (with $N$) error probability while using an efficient decoding algorithm. To see the dependence of the error probability on the number of tests, denote
\[
\epsilon = \beta\left(\frac{1}{2}\left(1-\frac{\ln 2}{K}\right)-\delta\right) - 1.
\]
Then, if the number of tests satisfies
\[
T \ge \frac{1+\epsilon}{\frac{1}{2}(1-\frac{\ln 2}{K})-\delta} K \log N
\]
one has
\[
P_e \leq N^{-\epsilon}.
\]
Thus, while the results under ML decoding (\Cref{direct theorem1}) show that any value of $\delta <1$ is possible (with a $\frac{1}{1-\delta}$ toll on $T$ compared to non-secure GT), the analysis herein suggests that using the efficient algorithm, one can have a small error probability only for $\delta < 1/2$, and the toll on $T$ is greater than $\frac{1}{\frac{1}{2}-\delta}$. \tifs{This is consistent with the fact that this algorithm is known to achieve only slightly more than half of the capacity for non-secure GT \cite{aldridge2014group}. Both these results may be due to an inherent deficiency in the algorithm.}
\begin{remark}[Complexity]
It is easy to see that the algorithm runs over all rows in the codebook, and compares each one to the vector of tests' results. The length of each row is $T$. There are $N$ items, each having about $2^{\frac{\delta}{K}T}$ rows in its bin. Since $T=O(K \log N)$, we have $O(N^2)$ rows in total. Thus, the number of operations is $O(N^2T)=O(KN^2 \log N)$. This should be compared to $O(KN\log N)$ without any secrecy constraint.
\end{remark}

\Cref{fig:DND_simulation} includes simulation results of the secure DND GT algorithm proposed, compared with ML decoding and the upper and lower bounds on the performance of ML.

 %
\section{Information Leakage at the Eavesdropper (Strong Secrecy)}\label{strong_secrecy}
\tifs{In this section, for any $K$ we wish to prove the strong security constraint is met, that is, $I(W;Z^{T})\rightarrow 0$. Denote by $\mathcal{C}_T$ the random codebook and by $\textbf{X}_{\mathcal{S}_w}^T$ the set of codewords corresponding to the true defective items. We assumed $W \in \{1,\ldots, {N \choose K}\}$ is uniformly distributed, that is, there is no \emph{a-priori} bias to any specific subset.
Further, the codebook includes independent and identically distributed codewords. The eavesdropper, observing $Z^T$, wishes to identify the identity of the $K$ defective items, one subset out of ${N \choose K}$ subsets.}

\tifs{To analyze the information leakage at the eavesdropper, we view the channel to Eve as a \emph{Multiple Access Channel, followed by a Binary Erasure Channel}. Specifically, each of the defective items can be considered as a ``user" with $M$ specific codewords. This is since each defective item ``transmits" a codeword, one out of $M$ (only these affect the output). Eve's goal, in SGT, is to identify the active users, the ones who actually transmitted. The channel to Eve is thus a binary Boolean MAC (summing the $K$ transmitted words), followed by a BEC$(1-\delta)$. \emph{The sum capacity to Eve is hence $\delta$}. By the channel coding theorem, this is a strict bound on the number of information bits Eve can get reliably. In fact, since the whole codebook is Bernoulli i.i.d.\ with an input distribution  $(\frac{\ln(2)}{K},1-\frac{\ln(2)}{K})$ for the testing matrix and uniform distribution on the items, we operate at \emph{an equal rate from each defective item}, hence Eve can obtain from each active user a rate of at most $\delta/K$. We will assume Eve indeed gets this amount of information from each item (user), i.e., we will view this as Eve obtaining a sub-matrix $\tilde{\textbf{Z}}$ of possibly transmitted codewords. On the one hand, it is a sub-matrix, since from each codeword (row) Eve achieves at most a fraction $\delta/K$. Yet, this only helps Eve, as we assume Eve's channel \emph{does not sum the codewords at the positions which Eve sees}. Thus, Eve will not be confused by the existence of other users. This way, we can analyze \emph{an upper bound on the information that leaks to Eve on the identity of the subset of defective items}\footnote{\tifs{Providing this information to Eve only makes her stronger and thus a coding scheme that keeps Eve ignorant, will also succeed against the original Eve.}}.}

\tifs{Thus, per codeword $X_j^T$, an upper bound on Eve's information is by assuming a fraction of at most $\frac{\delta}{K}$ of the $T$ entries are directly accessible to Eve (as if a genie revealed the summands for her).}
\tifs{Consider now Eve's side, observing a partial row of length $\frac{\delta}{K}T$. Eve wishes to find \emph{which rows in the original codebook are consistent with her observation}. If all rows consistent are within a bin of a certain item, Eve can identify this item as defective. However, we show that Eve's true status is much worse. In each bin, Eve will find many consistent rows, and, moreover, the number of consistent rows will be relatively similar for all bins (in a notion made precise below). To this end, remember that the input distribution is $(\frac{\ln(2)}{K},1-\frac{\ln(2)}{K})$, i.i.d.\ on all rows. Hence, per bit, the probability that Eve's observation is similar to a given row, which was also i.i.d.\ with the same distribution, is}
\begin{equation*}
    \tifs{\left(\frac{\ln(2)}{K}\right)^2+\left(1-\frac{\ln(2)}{K}\right)^2= 1 - \frac{2\ln(2)}{K} + 2\left(\frac{\ln(2)}{K}\right)^2.}
\end{equation*}
\tifs{Since Eve's observation is of length $\frac{\delta}{K}T$, the probability of a codeword being consistent with her observation is}
\begin{equation*}
    \tifs{\left(1 - \frac{2\ln(2)}{K} + 2\left(\frac{\ln(2)}{K}\right)^2\right)^{\frac{\delta}{K}T}.}
\end{equation*}
\tifs{As for each item we have $M =2^{\left(\frac{\delta+\epsilon^{\prime}}{K}\right)T}$ codewords, \emph{on average}, the  number of candidate codewords Eve has per item is}
\begin{multline*}
\tifs{\left(1 - \frac{2\ln(2)}{K} + 2\left(\frac{\ln(2)}{K}\right)^2\right)^{\frac{\delta}{K}T}2^{\left(\frac{\delta+\epsilon^{\prime}}{K}\right)T}}\\
\tifs{= 2^{\frac{T}{K}\left(\delta\Big( 1 + \log \big( 1 - \frac{2\ln(2)}{K} + 2\left(\frac{\ln(2)}{K}\right)^2\big)\Big)+\epsilon^{\prime}\right)}.}
\end{multline*}
\tifs{Denote this number by $2^{\frac{T}{K}\tilde{E}}$. Note that for any $K\geq 1$, 
\[
1 + \log \left( 1 - \frac{2\ln(2)}{K} + 2\left(\frac{\ln(2)}{K}\right)^2\right) \geq 0,
\]
hence $\tilde{E}\geq\epsilon^{\prime}$.}

\tifs{We now wish to show that not only Eve has (exponentially) many candidates per item, the probability that the number deviates from the average above is very small. Denote the set of candidates by $\mathcal{S}h(\tilde{Z}_{\tilde{j}}^{T})$, and further define the event}
\begin{equation*}
 \tifs{\mathcal{E}_{C_1}(\tilde{Z}_{\tilde{j}}^{T}):= \{ (1-\gamma)2^{\frac{T}{K}\tilde{E}}\leq
  |\mathcal{S}h(\tilde{Z}_{\tilde{j}}^{T})| \leq (1+\gamma)2^{\frac{T}{K}\tilde{E}}\}.}
\end{equation*}
\tifs{By the Chernoff bound}
\begin{equation}\label{eq. Chernoff}
  \tifs{P(\mathcal{E}_{C_1}(\tilde{Z}_{\tilde{j}}^{T})) \geq 1- 2^{-\frac{\gamma^2}{2}2^{\frac{T}{K}\tilde{E}}}.}
\end{equation}
\tifs{Note that \eqref{eq. Chernoff} applies for any $0 \leq \gamma \leq 1$, regardless of $T$. Yet, it is clear that for the result to be meaningful, $\gamma$ should be small enough on the one hand, yet not too small, so the exponent in \eqref{eq. Chernoff} is bounded away from zero. We will characterize $\gamma$ in the sequel.}

\tifs{We are now able to show that the mutual information is indeed negligible. Denote by $\textbf{1}_{\mathcal{E}_{C_1}}$ the indicator for the event where the actual number of rows \tifs{does not deviate} from the average. Therefore, under $\textbf{1}_{\mathcal{E}_{C_1}}$, Eve has between $(1-\gamma)2^{\frac{T}{K}\tilde{E}}$ and $ (1+\gamma)2^{\frac{T}{K}\tilde{E}}$ candidates per item, and thus:}
\tifs{\begin{align}\label{leakageS}
& I(W;Z^{T}) \nonumber\\
& = H(W)-H(W|Z^{T}) \nonumber\\
&\leq H(W)-H(W|Z^{T},\textbf{1}_{\mathcal{E}_{C_1}})\nonumber\\
&= H(W)-[P(\mathcal{E}_{C_1}^c)H(W|Z^{T},\textbf{1}_{\mathcal{E}_{C_1}} = 0)\nonumber\\
&\hspace{2.5cm} +P(\mathcal{E}_{C_1})H(W|Z^{T},\textbf{1}_{\mathcal{E}_{C_1}} = 1)]\nonumber\\
&\leq H(W)-P(\mathcal{E}_{C_1})H(W|Z^{T},\textbf{1}_{\mathcal{E}_{C_1}} = 1)\nonumber\\
&\leq \log \binom{N}{K}-\left( 1- 2^{-\frac{\gamma^2}{2}2^{\frac{T}{K}\tilde{E}}} \right)H(W|Z^{T},\textbf{1}_{\mathcal{E}_{C_1}}=1).
\end{align}}
\tifs{Finally, we consider $H(W|Z^{T},\textbf{1}_{\mathcal{E}_{C_1}}=1)$. Conditioned on $\textbf{1}_{\mathcal{E}_{C_1}}=1$, Eve has between $(1-\gamma)2^{\frac{T}{K}\tilde{E}}$ and $(1+\gamma)2^{\frac{T}{K}\tilde{E}}$ rows for each item. Hence, the mass she assigns each subset of possibly defective items is in the range $(1 \pm \gamma)K2^{\frac{T}{K}\tilde{E}}$, with a total mass of $\binom{N}{K}K2^{\frac{T}{K}\tilde{E}}$ (this is the total number of candidates, rows, she has for $K$ items). This results in assigning each subset a probability in the range $\frac{1}{\binom{N}{K}} \pm \frac{\gamma}{\binom{N}{K}}$. Computing the conditional entropy using this probability, and applying the Taylor series $-\log\left(\frac{1}{L}+\frac{\gamma}{L}\right) = -\log\left(\frac{1}{L}\right)-\gamma+O(\gamma^2)$, we have
\begin{equation*}
H(W|Z^{T},\textbf{1}_{\mathcal{E}_{C_1}}=1) = \log \binom{N}{K} - O\left(\gamma \log \binom{N}{K}\right).
\end{equation*}
Since $T \ge \log \binom{N}{K}$, setting $\gamma = \frac{1}{T^2}$ results in both making the $\gamma \log \binom{N}{K}$ term negligible with large $T$, and $2^{-\frac{\gamma^2}{2}2^{\frac{T}{K}\tilde{E}}} \to 0$ as $T \to \infty$, hence $I(W;Z^{T}) \to 0$ as $T \to \infty$, which completes the proof.
}

 %
\section{Scaling Regime of $K=o(N)$}\label{Kgrow}
Here, we point out that the result in \Cref{direct lemma1} can be extended to the more general case where both $N$ and $K$ are allowed to scale simultaneously such that $K=o(N)$. To attain order-optimal results as given \Cref{theorem_K=o(N)}, the scaling regime of $K$ is $O(\log(N))$. However, in the regime of $K=\omega(\log(N))$, reliability and secrecy in the SGT model can still be achieved, yet the result is no longer tight. E.g., when $K = N^{1/4}$, the reliability and secrecy constraints are met for $T=O\left(\frac{\sqrt N\log(N)}{1-\delta}\right)$. \tifs{Remembering the definition in \Cref{relia}, where $(\mathcal{S}^{1},\mathcal{S}^{2})$ denote a partition of the defective set $S$ into disjoint sets $\mathcal{S}^{1}$ and $\mathcal{S}^{2}$, and the codebook construction, $X_{\mathcal{S}^1}$ and $X_{\mathcal{S}^2}$ are independent, hence $P(X_{\mathcal{S}^2}) = P(X_{\mathcal{S}^2}| X_{\mathcal{S}^1})$. Hence, using the} analysis in the proof of \Cref{error lemma2} and similar techniques to \cite{aksoylar2014information,aksoylar2017sparse}, the error exponent in \Cref{error lemma2} can be written as
\ifdouble
\begin{multline}
E_{o}(\rho)= -\log \sum_{Y\in \{0,1\}}\sum_{X_{\mathcal{S}^2}} P(X_{\mathcal{S}^2})\\
\Bigg[\sum_{X_{\mathcal{S}^1}}P(X_{\mathcal{S}^1})P(Y|X_{\mathcal{S}^1},X_{\mathcal{S}^2})^{\frac{1}{1+\rho}}\Bigg]^{1+\rho},\text{ } 0 \leq\rho \leq1.\nonumber
\end{multline}
\else
\begin{equation*}
E_{o}(\rho)= -\log \sum_{Y\in \{0,1\}}\sum_{X_{\mathcal{S}^2}} P(X_{\mathcal{S}^2})
\Bigg[\sum_{X_{\mathcal{S}^1}}P(X_{\mathcal{S}^1})P(Y|X_{\mathcal{S}^1},X_{\mathcal{S}^2})^{\frac{1}{1+\rho}}\Bigg]^{1+\rho},\text{ } 0 \leq\rho \leq1.
\end{equation*}
\fi
The outcome $Y$ at the noiseless legitimate decoder depends on the Bernoulli random variables $\tilde{Y}_1=\bigvee_{d \in \mathcal{S}^1}X_{d}$ and $\tilde{Y}_2=\bigvee_{d \in \mathcal{S}^2}X_{d}$, with probability $p_1=1-(1-\ln(2)/K)^i$ and $p_2=1-(1-\ln(2)/K)^{K-i}$, respectively. Hence, substituting $\tilde{Y}_1,\tilde{Y}_2$ in the error exponent, where \tifs{$Q_i(\cdot)$ denotes the probability distribution, according to $p_i$, $i \in \{1,2\}$}, $P(Y|\tilde{Y}_1,\tilde{Y}_2) = 1$ if $Y = \tilde{Y}_1 \vee \tilde{Y}_2$  and $0$ otherwise, and since this is a binary function, we have
\ifdouble
\begin{multline}
E_{o}(\rho)= -\log \sum_{Y\in \{0,1\}}\sum_{\tilde{Y}_2 \in \{0,1\}} Q\tifs{_2}(\tilde{Y}_2)\\
\Bigg[\sum_{\tilde{Y}_1\in \{0,1\}}Q\tifs{_1}(\tilde{Y}_1)\textbf{1}\{Y = \tilde{Y}_1 \vee \tilde{Y}_2\}\Bigg]^{1+\rho},\text{ } 0 \leq\rho \leq1.\nonumber
\end{multline}
\else
\begin{equation*}
E_{o}(\rho)= -\log \sum_{Y\in \{0,1\}}\sum_{\tilde{Y}_2\in \{0,1\}} Q(\tilde{Y}_2) \Bigg[\sum_{\tilde{Y}_1\in \{0,1\}}Q(\tilde{Y}_1)
\textbf{1}\{Y = \tilde{Y}_1 \vee \tilde{Y}_2\}\Bigg]^{1+\rho},\text{ } 0 \leq\rho \leq1.
\end{equation*}
\fi
We now compute the realization of the error exponent for the two possible cases. When $Y=0$, we have $(1-p_2)(1-p_1)^{1+\rho}$, and when $Y=1$, we have $p_2+(1-p_2)p_1^{1+\rho}$, such that
\[
E_{o}(\rho)= -\log \Bigg(p_2+(1-p_2)\Bigg[p_{1}^{1+\rho}+(1-p_1)^{1+\rho}\Bigg]\Bigg).
\]
\tifs{Using the Taylor series for $e^{-\frac{\ln(2)}{K}}$ we have $p_1 = 1-e^{-i\ln(2)/K} + O\left(\frac{1}{K}\right)$, while $p_2 = 1-e^{-(K-i)\ln(2)/K} + O\left(\frac{1}{K}\right)$}. By choosing $\rho = 1$ and denoting $\alpha = i\ln(2)/K$, \tifs{for large enough $K$} we have
\ifdouble
\begin{multline*}
\hspace{-0.4cm} E_{o}(1) \tifs{=} -\log \Bigg(1-e^{-(\ln(2)-\alpha)} 
\\
\hspace{-0.5cm} + e^{-(\ln(2)-\alpha)}(1-e^{-\alpha})^2 +e^{-(\tifs{\ln(2)}+\alpha)} \tifs{+O\left(\frac{1}{K}\right)}\Bigg) 
\\
\hspace{4.7cm} = -\log \Bigg( e^{-\alpha} \tifs{+O\left(\frac{1}{K}\right)} \Bigg),
\end{multline*}
\else
\begin{equation*}
E_{o}(1)\cong -\log \Bigg(1-e^{-(\ln(2)-\alpha)} + e^{-(\ln(2)-\alpha)}(1-e^{-\alpha})^2+e^{-(\tifs{\ln(2)}+\alpha)}\Bigg) \cong -\log \Bigg( 1-\tifs{2e^{-\ln(2)}}+\frac{2}{e}e^{-\alpha}\Bigg).
\end{equation*}
\fi
Hence, $E_{o}(1)=\Theta(\alpha)$, for $\alpha=\Theta(1)$ and $\alpha=o(1)$.
For a large enough constant $c$ and $\rho=1$, we are now ready to show that $T\emph{f}(\rho)\rightarrow \infty$ when
\begin{equation}\label{eq:t_k=cn}
T = \frac{cK\log N}{1-\delta}.
\end{equation}
Using the union bound on the error probability $P(E_i)$ given in \Cref{error lemma2}, the conditional error probability $P_{e|w}$ is upper bounded by
\begin{multline*}
\hspace{-0.4cm} \sum_{i=1}^{K}P(E_i) \leq K \max_{i} P(E_i) \leq\max_{i}K \\
\exp\left(-T\left(E_{o}(\rho)-\rho\frac{\log\binom{N-K}{i}M^{i}}{T}-\frac{\log\binom{K}{i}}{T}-\frac{K}{T}\right)\right).
\end{multline*}
Hence, \tifs{given \eqref{eq:exp_bound}} following the definition of $T\emph{f}(\rho)$, we need to ensure that
\begin{equation*}
TE_{o}(\rho) \geq \rho\log\binom{N-K}{i}M^{i}+\log\binom{K}{i}+ K +\log K.
\end{equation*}
Where $\rho=1$, substituting \eqref{eq:t_k=cn} and
\[
\log_2(M^i)=iT\frac{\delta-\epsilon_K}{K}= \frac{ci\log N}{1-\delta} (\delta-\epsilon_K),
\]
we have,
\ifdouble
\begin{multline*}
T\emph{f}(1)=\frac{ci\log N}{1-\delta} - i\log \binom{N-K}{i}-\frac{ci\log N}{1-\delta}(\delta-\epsilon_K)\\
-i\log \binom{K}{i} - K - \log K\rightarrow \infty,
\end{multline*}
\else
\begin{equation*}
T\emph{f}(1)=\frac{ci\log N}{1-\delta} - i\log \binom{N-K}{i}-\frac{ci\log N}{1-\delta}(\delta-\epsilon_K)-i\log \binom{K}{i} - K - \log K\rightarrow \infty,
\end{equation*}
\fi
for a large enough constant $c$ (if $K=O(\log(N))$). Since $\log\binom{n}{i} \leq i\log \frac{ne}{i}$, increasing the constant $c$ slightly, we obtain the simpler
\ifdouble
\begin{multline*}
T\emph{f}(1)\geq c\Big(1+\frac{\epsilon_K}{1-\delta} \Big)i\log N\\ - i\log \frac{(N-K)e}{i}-i\log \frac{Ke}{i} - K - \log K.
\end{multline*}
\else
\begin{equation*}
T\emph{f}(1)\geq c\Big(1+\frac{\epsilon_K}{1-\delta} \Big)i\log N - i\log \frac{(N-K)e}{i}-i\log \frac{Ke}{i} - K - \log K.
\end{equation*}
\fi
From the above it is clear that when $K=O(\log(N))$ the result is achieved with some constant $c$. For a larger $K$, one has to choose a non-constant $c$, which grows with $K$, which will result in more tests. This is doable for any $K$, yet may not be tight for $K=\omega(\log(N))$.       %
\section{Conclusions}\label{conc}
\tifs{In this paper, we proposed a novel non-adaptive SGT scheme, with parameters $N,K$ and $T$, which is asymptotically \emph{reliable} and \emph{secure}. Specifically, with a new proposed test design,} when the fraction of tests observed by Eve is $0 \leq \delta <1$, we prove that the number of tests required for both correct reconstruction at the legitimate user (with high probability) and negligible mutual information at Eve's side\footref{note2} is $\frac{1}{1-\delta}$ times the number of tests required with no secrecy constraint.
We further provide sufficiency and necessity bounds on the number of tests required in the SGT model to obtains both, \emph{reliability} and \emph{secrecy} constraints.
Moreover, we analyze in the proposed secure model, computationally efficient algorithms at the legitimate decoder, previously considered for the non-secure GT in the literature which identify the definitely non-defective items.
\appendices
\tifs{\section{Proof of \Cref{error lemma2}}\label{appendix:APPENDIX A}}
\tifs{Since all subsets of size $K$ are equality likely to be defective, without loss of generality, we consider the error probability given that the first $K$ items are defective, that is, $S_{w=1}$ is the defective set. Denote this probability by $P_{e|1}$. We have $P_{e|1} \leq \sum_{i=1}^{K}P(E_i)$, where $E_i$ is the event of a decoding error in which the decoder declares a defective set which differs from the true one in exactly $i$ items.}

\tifs{In general, we follow the derivation in \cite{atia2012boolean}. However, there is a key difference. In the code construction suggested in Section IV, for each item there are $M$ possible codewords (a ``bin" of size $M$). Only one of these codewords is selected by the mixer to decide in which pool tests the item will participate. Thus, when viewing this problem as a channel coding problem, if an item is defective, one and only one codeword out of its bin is actually transmitted (and summed with the codewords of the other $K-1$ defective items). Since the decoder does not know which codewords were selected in each bin (the randomness is known only to the mixer), there are multiple error events to consider. E.g., events where the decoder choose the wrong codeword for some items, yet identified parts of the bins correctly, and, of course, events where the codeword selected was from a wrong bin. This complicates the error analysis. Moreover, we wish to employ the correction suggested in \cite{atia2015correction}, which results in a simpler yet stricter bound.}

\tifs{Consider the event $E_i$. $E_i$ can be broken down into two disjoint events. The first is $E_i$ \emph{and} the event that the codewords selected for the correct $K-i$ items are the true transmitted ones, and the second is the event of both $E_i$ \emph{and the event that at least one of the codewords selected for the correct items is wrong}. Denote the first event as $E'_i$. It will be used in the sequel. Now, consider the case where we have $E_i$, that is, a correct decision on $K-i$ items, yet, out of these $K-i$ items, the decoder identified only $j$ codewords right, $0 \leq j \leq K-i$, and for the rest, it identified a wrong codeword in the right bin. Let $L_{ij}$ denote this event, for all possible subsets of size $j$, and $L^\xi_{ij}$, denote a specific subset, out of the ${K-i \choose j}$ possible ones. We have:
\begin{eqnarray}\label{E and E'_new}
\sum_{i=1}^{K}P(E_i) 
&=&\sum_{i=1}^{K} \sum_{j=0}^{K-i} P(L_{ij}) 
\nonumber\\
&=& \sum_{i=1}^{K} \sum_{j=0}^{K-i} \sum_{\xi} P(L^\xi_{ij}) \nonumber\\
&=&\sum_{i=1}^{K} \sum_{j=0}^{K-i} {K-i \choose j} P(L^\xi_{ij})\nonumber\\
&\leq& \sum_{i=1}^{K} \sum_{j=0}^{K-i} {K-i \choose j}P(E'_{K-j}),
\end{eqnarray}
were the last inequality is since both $E'_{K-j}$ and $L^\xi_{ij}$ refer to decoding $j$ codewords correctly, yet $E'_{K-j}$ alows for the wrong ones to be in a less restrictive set of bins compared to $L^\xi_{ij}$ (remember that $K \ll N$).}
\tifs{Continuing by interchanging $i$ and $j$, we have
\begin{eqnarray}
P_{e|1} &\leq& \sum_{j=0}^{K-1} \sum_{i=1}^{K-j} {K-i \choose j}P(E'_{K-j})\nonumber\\
&=& \sum_{j=0}^{K-1} P(E'_{K-j}) \sum_{i=1}^{K-j} {K-i \choose j} \nonumber\\
&=& \sum_{j=0}^{K-1} P(E'_{K-j}) \sum_{l=0}^{K-j-1} {j+l \choose j} \nonumber\\
&\stackrel{(a)}{=}& \sum_{j=0}^{K-1} P(E'_{K-j}) {K \choose j+1}
\nonumber\\
&\stackrel{b}{\leq}& 2^K \sum_{j=0}^{K-1} P(E'_{K-j})  \nonumber\\
&=& 2^K \sum_{i=1}^{K}P(E'_i),
\end{eqnarray}
were (a) follows from the Rising Sum of Binomial Coefficients, $\sum_{i=0}^{m} {n+i \choose n} = {n+m+1 \choose n+1}$, and (b) follows from the Binomial upper bound $\forall$  $0 \leq j+1 \leq K$, thus, ${K \choose j+1} \leq 2^K$.}

\tifs{We now bound $P(E'_i)$. Particularly, we will establish the following lemma.}
\begin{lemma}\label{error lemma2 with E'}
\tifs{The error probability $P(E'_{i})$ is bounded by}
\begin{equation*}
\tifs{P(E'_{i}) \leq 2^{-T\left(E_{o}(\rho)-\rho\frac{\log\binom{N-K}{i}M^i}{T}-\frac{\log\binom{K}{i}}{T}\right)},}
\end{equation*}
\tifs{where the error exponent $E_{o}(\rho)$ is given by}
\ifdouble
\begin{multline}
\tifs{E_{o}(\rho)= -\log \sum_{Y\in \{0,1\}}\sum_{X_{\mathcal{S}^2}\in \{0,1\}}\Bigg[\sum_{X_{\mathcal{S}^1}\in \{0,1\}}P(X_{\mathcal{S}^1})}\\
\tifs{p(Y,X_{\mathcal{S}^2}|X_{\mathcal{S}^1})^{\frac{1}{1+\rho}}\Bigg]^{1+\rho},\text{ } 0 \leq\rho \leq1.\nonumber}
\end{multline}
\else
\begin{equation*}
 \tifs{E_{o}(\rho)= -\log \sum_{Y\in \{0,1\}}\sum_{X_{\mathcal{S}^2}\in \{0,1\}}\Bigg[\sum_{X_{\mathcal{S}^1}\in \{0,1\}}P(X_{\mathcal{S}^1}) p(Y,X_{\mathcal{S}^2}|X_{\mathcal{S}^1})^{\frac{1}{1+\rho}}\Bigg]^{1+\rho},\text{ } 0 \leq\rho \leq1.}
\end{equation*}
\fi
\end{lemma}
\begin{proof}
\tifs{Denote by $\mathcal{A} = \{w \in \mathcal{W} : |S_{1^{c},w}|=i,|S_{w}|=K\}$
the set of indices corresponding to sets of $K$ items that differ from the true defective set $S_{1}$ in exactly $i$ items. Using the same notation in \cite{atia2012boolean}, $S_{1^{c},w}$ denotes the set of items which are in $S_w$ but not in $S_1$. Considering $E'_i$, we assume the decoder not only got $K-i$ items right, but also the correct codeword in each such bin. Thus, this is exactly the setup in \cite[eq. (25)]{atia2012boolean}, and for all $s>0$ we have 
\ifdouble
\begin{multline}\label{weaker_error2}
P[E'_i|w_0=1, \textbf{X}_{\mathcal{S}_1},Y^T]
\\
 \leq  \sum_{\mathcal{S}_{1,w}} \sum_{\mathcal{S}_{1^{c},w}}\sum_{\textbf{X}_{\mathcal{S}_{1^{c},w}}}
P({\scriptstyle \textbf{X}_{\mathcal{S}_{1^{c},w}}})\frac{{\scriptstyle p_{w}(Y^T,\textbf{X}_{\mathcal{S}_{1,w}}|\textbf{X}_{\mathcal{S}_{1^{c},w}})^s}} {{\scriptstyle p_{1}(Y^T,\textbf{X}_{\mathcal{S}_{1,w}}|\textbf{X}_{\mathcal{S}_{1,w^{c}}})^s}}\nonumber.
\end{multline}
\else
\begin{equation}\label{weaker_error2}
P[E'_i|w_0=1, \textbf{X}_{\mathcal{S}_1},Y^T] 
 \leq \sum_{\mathcal{S}_{1,w}} \sum_{\mathcal{S}_{1^{c},w}}\sum_{\textbf{X}_{\mathcal{S}_{1^{c},w}}}
P({\scriptstyle \textbf{X}_{\mathcal{S}_{1^{c},w}}})\frac{{\scriptstyle p_{w}(Y^T,\textbf{X}_{\mathcal{S}_{1,w}}|\textbf{X}_{\mathcal{S}_{1^{c},w}})^s}} {{\scriptstyle p_{1}(Y^T,\textbf{X}_{\mathcal{S}_{1,w}}|\textbf{X}_{\mathcal{S}_{1,w^{c}}})^s}}\nonumber.
\end{equation}
\fi
Continuing with the Gallager-type bound, for any $0\leq\rho\leq 1$,
\ifdouble
\begin{eqnarray}
&&\hspace{-0.5cm} P[E_i|w_0=1, \textbf{X}_{\mathcal{S}_1},Y^T]\nonumber\\
&\hspace{-0.5cm} \stackrel{(a)}{\leq} &\hspace{-0.4cm} \Big(\sum_{\mathcal{S}_{1,w}} \sum_{\mathcal{S}_{1^{c},w}}\sum_{\textbf{X}_{\mathcal{S}_{1^{c},w}}}
P({\scriptstyle\textbf{X}_{\mathcal{S}_{1^{c},w}}}) \frac{{\scriptstyle p_{w}(Y^T,\textbf{X}_{\mathcal{S}_{1,w}}|\textbf{X}_{\mathcal{S}_{1^{c},w}})^s}} {{\scriptstyle p_{1}(Y^T,\textbf{X}_{\mathcal{S}_{1,w}}|\textbf{X}_{\mathcal{S}_{1,w^{c}}})^s}}\Big)^{\rho}\nonumber\\
&\hspace{-0.5cm} \stackrel{(b)}{\leq} &\hspace{-0.4cm} \Big(\sum_{\mathcal{S}_{1,w}} {\scriptstyle\binom{N-K}{i}M^{i}}\sum_{\textbf{X}_{\mathcal{S}_{1^{c},w}}}
P({\scriptstyle\textbf{X}_{\mathcal{S}_{1^{c},w}}})\frac{{\scriptstyle p_{w}(Y^T,\textbf{X}_{\mathcal{S}_{1,w}}|\textbf{X}_{\mathcal{S}_{1^{c},w}})^s}} {{\scriptstyle p_{1}(Y^T,\textbf{X}_{\mathcal{S}_{1,w}}|\textbf{X}_{\mathcal{S}_{1,w^{c}}})^s}}\Big)^{\rho}\nonumber\\
&\hspace{-0.5cm} \stackrel{(c)}{\leq} &\hspace{-0.4cm} {\scriptstyle\binom{N-K}{i}^\rho M^{i\rho}}\sum_{\mathcal{S}_{1,w}}\Big(\sum_{\textbf{X}_{\mathcal{S}_{1^{c},w}}}
P({\scriptstyle\textbf{X}_{\mathcal{S}_{1^{c},w}}})\frac{{\scriptstyle p_{w}(Y^T,\textbf{X}_{\mathcal{S}_{1,w}}|\textbf{X}_{\mathcal{S}_{1^{c},w}})^s}} {{\scriptstyle p_{1}(Y^T,\textbf{X}_{\mathcal{S}_{1,w}}|\textbf{X}_{\mathcal{S}_{1,w^{c}}})^s}}\Big)^{\rho}\nonumber.
\end{eqnarray}
\else
\begin{eqnarray}
P[E_i|w_0=1, \textbf{X}_{\mathcal{S}_1},Y^T]&\stackrel{(a)}{\leq} & \Big(\sum_{\mathcal{S}_{1,w}} \sum_{\mathcal{S}_{1^{c},w}}\sum_{\textbf{X}_{\mathcal{S}_{1^{c},w}}}
P({\scriptstyle\textbf{X}_{\mathcal{S}_{1^{c},w}}}) \frac{{\scriptstyle p_{w}(Y^T,\textbf{X}_{\mathcal{S}_{1,w}}|\textbf{X}_{\mathcal{S}_{1^{c},w}})^s}} {{\scriptstyle p_{1}(Y^T,\textbf{X}_{\mathcal{S}_{1,w}}|\textbf{X}_{\mathcal{S}_{1,w^{c}}})^s}}\Big)^{\rho}\nonumber\\
& \stackrel{(b)}{\leq} & \Big(\sum_{\mathcal{S}_{1,w}} {\scriptstyle\binom{N-K}{i}M^{i}}\sum_{\textbf{X}_{\mathcal{S}_{1^{c},w}}}
P({\scriptstyle\textbf{X}_{\mathcal{S}_{1^{c},w}}})\frac{{\scriptstyle p_{w}(Y^T,\textbf{X}_{\mathcal{S}_{1,w}}|\textbf{X}_{\mathcal{S}_{1^{c},w}})^s}} {{\scriptstyle p_{1}(Y^T,\textbf{X}_{\mathcal{S}_{1,w}}|\textbf{X}_{\mathcal{S}_{1,w^{c}}})^s}}\Big)^{\rho}\nonumber\\
&\stackrel{(c)}{\leq} & {\scriptstyle\binom{N-K}{i}^\rho M^{i\rho}}\sum_{\mathcal{S}_{1,w}}\Big(\sum_{\textbf{X}_{\mathcal{S}_{1^{c},w}}}
P({\scriptstyle\textbf{X}_{\mathcal{S}_{1^{c},w}}})\frac{{\scriptstyle p_{w}(Y^T,\textbf{X}_{\mathcal{S}_{1,w}}|\textbf{X}_{\mathcal{S}_{1^{c},w}})^s}} {{\scriptstyle p_{1}(Y^T,\textbf{X}_{\mathcal{S}_{1,w}}|\textbf{X}_{\mathcal{S}_{1,w^{c}}})^s}}\Big)^{\rho}\nonumber.
\end{eqnarray}
\fi
(a) is since for any probability $p$ which is upper bounded by some $U$, $p \leq p^{\rho} \leq U^\rho$. (b) follows from the symmetry of the codebook and its binning structure. There are exactly $\binom{N-K}{i}M^{i}$ possible sets of codewords to consider for $\mathcal{S}_{1^{c},w}$ and all are equiprobable. (c) follows as the sum of positive numbers raised to the $\rho$-th power is smaller than the sum of the $\rho$-th powers.}

\tifs{We now follow the steps in \cite{atia2012boolean}, first substituting the conditional error probability in a summation over all codewords and output vectors, then noting there are $\binom{K}{K-i}$ sets $\mathcal{S}_{1,w}$, and the summation is independent on the actual set, and finally using the memoryless structure of the codebook and the channel:
\begin{eqnarray*}
 P(E'_i) &=& \sum_{\textbf{X}_{\mathcal{S}_{1}}}\sum_{Y^T}p_1(\textbf{X}_{\mathcal{S}^1},Y^T)P[E'_i|w_0=1, \textbf{X}_{\mathcal{S}_1},Y^T]\nonumber\\
&\leq& {\scriptstyle\binom{N-K}{i}^\rho M^{i\rho}}\sum_{\mathcal{S}_{1,w}}\sum_{Y^T}\sum_{\textbf{X}_{\mathcal{S}_{1}}}p_1({\scriptstyle \textbf{X}_{\mathcal{S}^1},Y^T})\nonumber\\
&& \left(\sum_{\textbf{X}_{\mathcal{S}_{1^{c},w}}}P({\scriptstyle\textbf{X}_{\mathcal{S}_{1^{c},w}}})
\frac{{\scriptstyle p_{w}(Y^T,\textbf{X}_{\mathcal{S}_{1,w}}|\textbf{X}_{\mathcal{S}_{1^{c},w}}})^s} {{\scriptstyle p_{1}(Y^T,\textbf{X}_{\mathcal{S}_{1,w}}|\textbf{X}_{\mathcal{S}_{1,w^{c}}}})^s}\right)^{\rho}\nonumber
\end{eqnarray*}
\begin{eqnarray*}
& \leq &
{\scriptstyle\binom{N-K}{i}^\rho M^{i\rho}\binom{K}{i}}\sum_{Y^T}\sum_{\textbf{X}_{\mathcal{S}_{1}}}p_1({\scriptstyle\textbf{X}_{\mathcal{S}^1},Y^T})\nonumber\\ 
&&\left(\sum_{\textbf{X}_{\mathcal{S}_{1^{c},w}}}P({\scriptstyle{\scriptstyle\textbf{X}_{\mathcal{S}_{1^{c},w}}}}) \frac{{\scriptstyle p_{w}(Y^T,\textbf{X}_{\mathcal{S}_{1,w}}|\textbf{X}_{\mathcal{S}_{1^{c},w}})^s}} {{\scriptstyle p_{1}(Y^T,\textbf{X}_{\mathcal{S}_{1,w}}|\textbf{X}_{\mathcal{S}_{1,w^{c}}})^s}}\right)^{\rho}
\nonumber\\
&=&
{\scriptstyle\binom{N-K}{i}^\rho M^{i\rho}\binom{K}{i}}\Big[\sum_{Y}\sum_{X_{\mathcal{S}_{1,w}}}
\nonumber\\
&& \hspace{-0.2cm}\big(\sum_{X_{\mathcal{S}_{1,w^{c}}}}P({\scriptstyle X_{\mathcal{S}_{1^{c},w}}}) p_1^{\frac{1}{1+\rho}}({\scriptstyle X_{\mathcal{S}_{1,w}},Y|X_{\mathcal{S}_{1,w^{c}}}})\big)^{1+\rho}\Big]^T
\nonumber\\
&=&
2^{-T\left(E_0(\rho) -\rho\left(\frac{\log{N-K \choose i}}{T}+\frac{i\log M}{T} \right) -\frac{\log{K \choose i}}{T}\right)}.
\end{eqnarray*}}
\end{proof}
             %
\begin{center}
  \huge Supplementary Materials
\end{center}
\appendices
\section*{Applications}\label{applications}
Here, we show several common applications which exemplify the applicability of the suggested code for GT to a diverse range of protocols and applications.

\subsection{Blood Testing}
The first application we suggest follows the traditional protocol where one wants to identify a small set of infected individuals out of a large population, utilizing blood tests \cite{dorfman1943detection}. However, in contrast to the traditional application, one wants to keep the identity of the infected as well as the healthy confidential such that no one, including the lab that examines the blood samples, will get any information regarding any of the examined individuals.  For this application we suggest two levels of privacy. In the first one the patients confide and are willing to disclose the results to the nurse that extracts the blood tests, yet wants to conceal the information from the lab that performs the examination. In the second scheme the patient wants to conceal the information also from the nurse that extracts the blood tests.

In the first application, a nurse collects the blood samples from the patients and prepare the pool tests, i.e., inserts a sample from each patient’s blood sample in a selected subset of $T$ test tubes. The selection on which test tubes to insert each patient’s blood samples is according to the SGT testing matrix given in Section IV. In the second stage, the nurse divides the $T$ tubes to $1/\delta$ sets, and sends each set to a distinct lab. In third stage, each lab examines each of the test tubes it received and identifies which of the tubes is positive. Each lab returns the results (the list of positive tubes) to the doctor’s clinic, which can identify the infected patients. Note that since the nurse utilizes the testing matrix as suggested in Section IV and ensures that each lab will get at most $T\delta$ test tubes, each lab is kept ignorant with respect to the status of all patient, as was proven in Section IV-B. On the other hand, the clinic which receives the results from all the labs, i.e., retains all the $T$ test tubes’ outcomes, can identify the infected patients. \tifs{In this application, the decoding procedure in the clinic may be either ML (Section IV) or DND (Section VII). Yet, $T$ needs to satisfy the conditions in Theorem 1 or Theorem 4, respectively.} An illustration of the suggested procedure is given in \Cref{figure:SBT}.

In the second procedure which we term ``Private model", the patients want to retain confidentiality from everyone, including the nurse, labs, doctor, clinic and other patients. In this model, after the nurse extracts the blood sample, the nurse \tifs{returns the} blood tube to the patient. In the clinic there are $T$ test tubes. Each patient randomly chooses from a discrete pool a list of test tubes in which he/she needs to put a sample of its blood in. Each such list is associated with a row in the secure testing matrix $\textbf{X}_{j}$, suggested in Section IV. Note that each list (row in the secure testing matrix) appears only once in the pool, such that only one patient can get this list. The patient keeps the list of its selected test tubes for his or her record. The other phases are exactly as before, the clinic sends the test tubes to $1/\delta$ different Labs, which reply which of the examined test tubes is positive. The clinic display the list of all infected tubes from all the labs. Each patient checks, if all the tubes in his or her list are positive, then he or she is infected otherwise he or she is healthy. Note that the “decoding” procedure each patient performs is according to the secure DND algorithm suggested in Section VII. \tifs{Hence, in this application, $T$ needs to satisfy the condition in Theorem 4.} Further note that the ``Private model" can also be attained by the nurse performing the mixtures, yet instead of holding the identity of the patients the nurse just provides each patient with a reference number. The list of the infected patients’ reference numbers is posted such that neither the nurse or the clinic need to know the identity of the examined patients.

\begin{figure}
  \centering
  \includegraphics[trim=0cm 0cm 0cm 0cm,clip,scale=0.85]{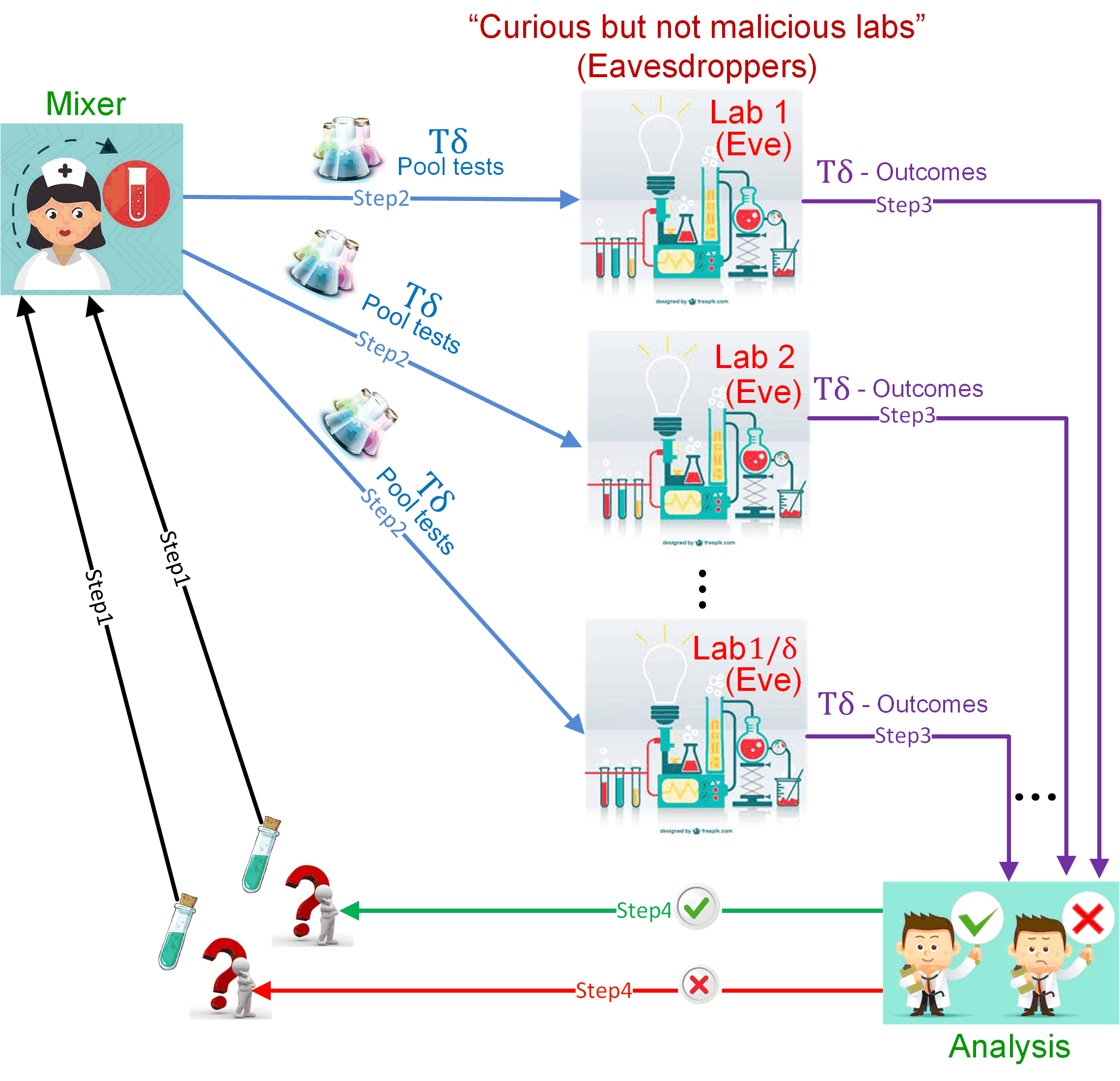}
  \caption{Secure Blood Testing.}
  \label{figure:SBT}
\end{figure}

\subsection{Wireless Sensors Network}\label{applications_wsn}
The next application relates to the prevalent Wireless Sensors Networks (WSN) and Internet of Things (IoT) Networks. Specifically, it relates to the setup in which a sink needs to collect reports from a large population of simple devices, yet only a small subset of devices is expected to have reports for transmission at a time, and the information transmitted is confidential and needs to be concealed from potential eavesdroppers.

In \cite{cohen2018secured} we suggest a highly efficient secured data gathering protocol for IoT networks using the concept of the secure GT code suggested herein. In this network we assume that there is large population of devices ($N$ devices), each with a bank of up to $C$ messages (the messages can be the same or different among different devices). The suggested protocol was designed to enable data gathering from $K$ devices simultaneously, such that the sink will be able to decode the information sent, yet an eavesdropper which observes only a noisy version of the channel output cannot gain any information on the transmitted messages or even on the identity of the nodes that sent these messages. In the suggested protocol each row from the SGT testing matrix given in Section IV represents a codeword. Accordingly, each message is assigned a different codeword (same message at different devices will be assigned different codewords). Whenever the sink broadcasts a predefined beacon which initiates a $T$ minislot interval. Each of the $K$ active devices omits energy on each of the minislots corresponding to the codeword it wishes to transmit. The sink only needs to identify the minislots in which some energy was detected in order to identify the message sent as well as the identity of the sender. Note that since each active device can send only one out of $C$ messages, and assuming that the eavesdropper can observe only a $\delta$ fraction of the minislots, a codeword of size $T = \Theta \left(\frac{K \log N C}{1-\delta}\right)$ ensures that the sink is able to decode the information while the eavesdropper cannot identify the message sent nor its sender. An illustration of the setup is given on \Cref{figure:WSN}. Note that the same application can be modified to work in the frequency domain rather than the time domain, i.e., instead of time slots each device chooses sub-channels (frequency bands) to omit the energy, such that the eavesdropper that can monitor only a subpart of the frequency band will be kept ignorant.

Another related application which was considered in many networks and algorithms is neighbor discovery \cite{mcglynn2001birthday,borbash2007asynchronous,angelosante2010neighbor,angelosante2007simple}. Neighbor discovery relates to the challenge in which in a dense network with large population of $N$ sensors, a node wishes to identify its $K$ neighbors (we assume that the number of neighbors is exactly $K$), e.g., \cite{luo2008neighbor}. Obviously a traditional Medium Access Control (MAC) protocol can apply. However, such protocols require contention and collision resolution mechanisms which can result in poor performance. Utilizing GT can dramatically improve the performance (e.g., \cite{luo2008neighbor}).  Specifically, similar to the previous example, by providing each device with a unique sequence (codeword/row in the codebook/testing-matrix) known to all the nodes in the network, each node can discover all its neighbors simultaneously (as explained in the previous example). Furthermore, if the codewords are taken from the SGT testing matrix given in Section IV, an eavesdropper which observes only a noisy version of the transmitted data will not be able to identify the Identity of the neighbors.

Moreover, in \cite{bajwa2007joint} proposed a joint source channel communication for distributed estimation in sensor networks which in a sense may be considered as a GT problem as well.

\begin{figure}
  \centering
  \includegraphics[trim=0cm 0cm 0cm 0cm,clip,scale=0.3]{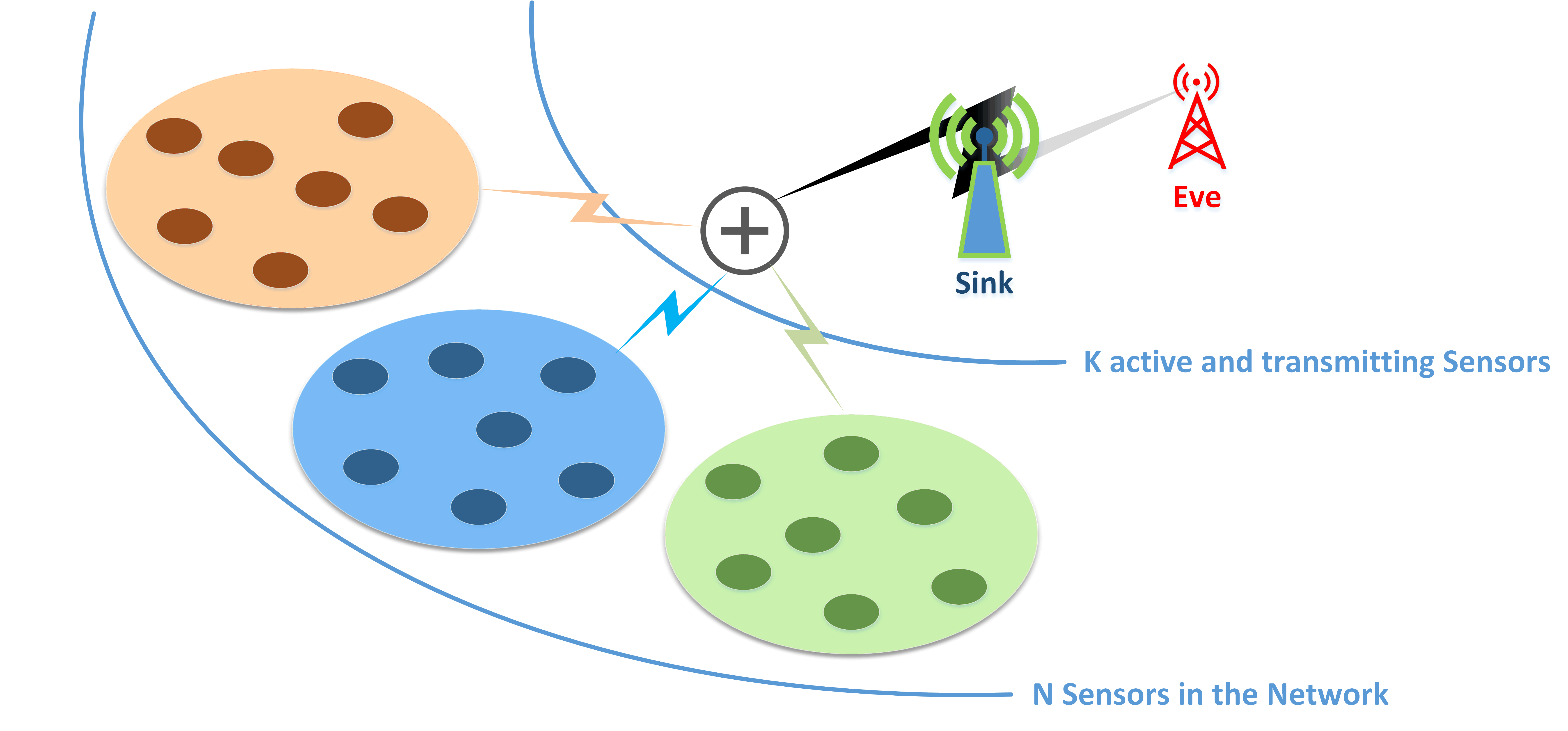}
  \caption{Secured data gathering protocol and neighbor discovery for IoT networks. We consider the case where there is a large population of $N$ sensors in the network and for neighbor discovery an upper bound of $K$ sensors in each neighborhood applies. For data aggregation protocol we assume $K$ active sensors simultaneously transmitting their messages to the sink over the whole wireless sensor network.}
  \label{figure:WSN}
\end{figure}

In this Section we introduced two applications that can utilize the suggested mechanism, however many other applications can utilize the secured GT suggested in this paper. Most applications that rely on group testing and wish to complement it with security or privacy layer can adopt the secured version suggested herein. For example, anomaly detection is an important task in processing large-scale data sets aimed at identifying unusual or rare events based on a large set of observations \cite{chandola2009anomaly}. Anomaly detection is essential in diverse domains such as intrusion detection, fraud detection, fault detection, system health monitoring, event detection in sensor networks, and detecting ecosystem disturbances. GT was utilized as one of the necessary components in various anomaly detection procedures (e.g., \cite{du1999combinatorial,tsai2004testing,muthukrishnan2005data,bai2007adaptive}). For example, in \cite{gilbert2001quicksand,gilbert2002fast,cormode2005s1,cormode2005s} GT was utilized to detect anomalies and possible attacks on wide networks such as the Internet. These procedures rely on monitoring traffic on many servers across the network, collect data and analyze at a central location. Since the amount of data needed to be collected is enormous, GT is utilized to collect an aggregated version of the data which is further merged along the route to the destination. Obviously, in many such occasions, the traversed data needs to remain confidential from potential eavesdropper and even from legitimate entities along the route. In such a scenario, the SGT version suggested herein can be utilized. Several other applications that can utilize the SGT suggested in this paper can be found in \cite{du1999combinatorial,kwang2006pooling} for a variety of fields and domains.
  %
\section*{Background}\label{BooleanCompressed}
Here, we present an extensive survey of the literature on group testing and security schemes.
\subsection{Group-testing}
Group-testing comes in various flavours, and the literature on these is vast. At the risk of leaving out much, we reference here just some of the models that have been considered in the literature, and specify our focus in this work.
\subsubsection{Performance Bounds}
GT can be \emph{non-adaptive}, where the testing matrix is designed beforehand, {\it adaptive}, where each new test can be designed while taking into account previous test \emph{results}, or a combination of the two, where testing is adaptive, yet with batches of non-adaptive tests. It is also important to distinguish between \emph{exact recovery} and a \emph{vanishing probability of error}.

To date, the best known lower bound on the number of tests required (non-adaptive, exact recovery) is $\Omega(\frac{K^2}{\log K}\log N)$ \cite{furedi1996onr}. The best known explicit constructions were given in \cite{porat2011explicit}, resulting in $O(K^2 \log N)$. However, focusing on exact recovery requires more tests, and forces a combinatorial nature on the problem. Settling for high probability reconstructions allows one to reduce the number of tests to the order of $K\log N$.\footnote{A simple information theoretic argument explains a lower bound. There are $K$ defectives out of $N$ items, hence $N \choose K$ possibilities to cover: $\log {N \choose K}$ bits of information. Since each test carries at most one bit, this is the amount of tests required. Stirling's approximation easily shows that for $K \ll N$, the leading factor of that is $K\log(N/K)$.} For example, see the channel-coding analogy given in \cite{atia2012boolean}. A similar analogy to wiretap channels will be at the basis of this work as well. In fact, probabilistic methods with an error probability guarantee appeared in \cite{sebHo1985two}, without explicitly mentioning GT, yet showed the $O(K \log N)$ bound. Additional probabilistic methods can be found in \cite{scarlett2015limits} for \emph{support recovery}, or in \cite{scarlett2016phase}, when an interesting \emph{phase transition} phenomenon was observed, yielding tight results on the threshold (in terms of the number of tests) between the error probability approaching one or vanishing. \tifs{Finally, \cite{johnson2018performance} propose a constant-column design that is superior to the Bernoulli design, i.e., where each individual joins an equal number of tests chosen uniformly at random without replacement.}
\subsubsection{A Channel Coding Interpretation}
As mentioned, the analogy to channel coding has proved useful \cite{atia2012boolean}. \cite{baldassini2013capacity} defined the notion of \emph{group testing capacity}, that is, the value of $\lim_{N \to \infty} \frac{\log{{N}\choose{K}}}{T}$ under which reliable algorithms exist, yet, over which, no reliable reconstruction is possible. A converse result for the Bernoulli, non-adaptive case was given in \cite{aldridge2017capacity}. Strong converse results were given in \cite{tan2014strong,johnson2015strong}, again, building on the channel coding analogy, as well as converses for noisy GT \cite{scarlett2016converse}. In \cite{aldridge2012adaptive}, adaptive GT was analyzed as a channel coding \emph{with feedback} problem.
\subsubsection{Efficient Algorithms}
A wide variety of techniques were used to design efficient GT decoders. Results and surveys for early non-adaptive decoding algorithms were given in \cite{chen2008survey,de2005optimal,indyk2010efficiently}. Moreover, although most of the works described above mainly targeted fundamental limits, some give efficient algorithms as well. In the context of this work, it is important to mention the recent COMP \cite{chan2014non}, DD and SCOMP \cite{aldridge2014group} algorithms, concepts from which we will use herein. \tifs{Note that in \cite{coja2019information}, it was shown that SCOMP is no better than DD and can, therefore, be discarded. Moreover, an efficient algorithm proposed in \cite{coja2019optimal} called SPIV was found that attains the information-theoretic lower bound.}

\subsection{Secure communication}
It is very important to note that {\it making GT secure is different from making communication secure, as remarked in Section I}. Now, we briefly survey the literature in secure communication, since many of the ideas/models/primitives in secure communication will have analogues in secure group-testing.
\subsubsection{Information-theoretic secrecy}
In a secure communication setting, transmitter Alice wishes to send a message $m$ to receiver Bob. To do so, she is allowed to encode $m$ into a (potentially random) function $x = f(m)$, and transmit $x$ over a medium. It is desired that the eavesdropper Eve should glean no information about $m$ from its (potentially noisy) observation $z$. This information leakage is typically measured via the mutual information between $m$ and $z$. The receiver Bob should be able to reconstruct $m$ based on its (also potentially noisy) observation of $x$ (and, potentially, a shared secret that both Bob and Alice know, but Eve is ignorant of).

There are a variety of schemes in the literature for information-theoretically secure communications.\footnote{Security in general has many connotations — for instance, in the information-theory literature it can also mean a scheme that is resilient to an {\it active adversary}, for instance a communication scheme that is resilient to jamming against a malicious jammer. In this work we focus our attention on {\it passive eavesdropping adversaries}, and aim to ensure secrecy of communications vis-a-vis such adversaries. We shall thus henceforth use the terms security and secrecy interchangeably.}
Such schemes typically make one of several assumptions (or combinations of these):

\begin{itemize}
\item {\it Shared secrets/Common randomness/Symmetric-key encryption:} The first scheme guaranteed to provide information-theoretic secrecy was by \cite{C1}, who analyzed the secrecy of one-time pad schemes and showed that they ensure perfect secrecy (no leakage of transmitted message). This scheme also provided lower bounds on the size of this shared key. The primary disadvantage of such schemes is that they require a shared key that is essentially as large as the amount of information to be conveyed, and it be continually refreshed for each new communication. These requirements typically make such schemes untenable in practice.
\item {\it Wiretap secrecy/Physical-layer secrecy:} Wyner \emph{et al.} \cite{C2,ozarow1984wire} first considered certain communication models in which the communication channel from Alice to Eve is a degraded (noisier) version of the channel from Alice to Bob, and derived the information-theoretic capacity for communication in such settings. These results have been generalized in a variety of directions. See \cite{csiszar2004secrecy,csiszar2008secrecy,C13} for (relatively) recent results. The primary disadvantage of such schemes is that they require that it be possible to instantiate communication channels from Alice to Bob that are better than the communication channel from Alice to Eve. Further, they require that the channel parameters of both channels be relatively well known to Alice and Bob, since the choice of communication rate depends on these parameters. These assumptions make such schemes also untenable in practice, since on one hand Eve may deliberately situate herself to have a relatively clearer view of Alice's transmission than Bob, and on the other hand there are often no clear physically-motivated reasons for Alice and Bob to know the channel parameters of the channel to Eve.
\item {\it Public discussion/Public feedback:} A notable result by Maurer (~\cite{maurer1993secret} and subsequent work - see \cite{C13} for details) significantly alleviated at least one of the charges level against physical-layer security systems, that they required the channel to Bob to be ``better" than the channel to Eve. Maurer demonstrated that feedback (even public feedback that is noiselessly observable by Eve) and multi-round communication schemes can allow for information-theoretically secure communication from Alice to Bob even if the channel from Alice to Bob is worse than the channel from Alice to Eve. Nonetheless, such public discussion schemes {\it still} require some level of knowledge of the channel parameters of the channel to Eve.
\end{itemize}
\subsubsection{Cryptographic security}
Due to the shortcomings highlighted above, modern communication systems usually back off from demanding information-theoretic security, and instead attempt to instantiate {\it computational security}. In these settings, instead of demanding small information leakage to arbitrary eavesdroppers, one instead assumes bounds on the computational power of the eavesdropper (for instance, that it cannot computationally efficiently invert ``one-way functions"). Under such assumptions one is then often able to provide conditional security, for instance with a public-key infrastructure ~\cite{kocher1996timing,stinson2005cryptography}.
Such schemes have their own challenges to instantiate. For one, the computational assumptions they rely on are sometimes unfounded and hence sometimes turn out to be prone to attack \cite{tagkey2004391,C13,dent2006fundamental}. For another, the computational burden of implementing cryptographic primitives with strong guarantees can be somewhat high for Alice and Bob~\cite{bernstein2008attacking}.

\subsection{Secure Group-Testing}
On the face of it, the connection between secure communication and secure group-testing is perhaps not obvious. We highlight below scenarios that make these connections explicit.
Paralleling the classification of secure communication schemes above, one can also conceive of a corresponding classification of secure GT schemes.
\subsubsection{Information-theoretic schemes}\label{Information-theoretic schemes}
\begin{itemize}
\item {\it Shared secrets/Common randomness/Symmetric-key encryption:} A possible scheme to achieve secure group testing, is to utilize a shared key between Alice and Bob. For example, consider a scenario in which Alice the nurse has a large number of blood samples that need to be tested for the presence of a disease. She sends them to a lab named Eve to be tested. To minimize the number of tests done via the lab, she pools blood samples appropriately. However, while the lab itself will perform the tests honestly, it can't be trusted to keep medical records secure, and so Alice keeps secret the identity of the people tested in each pool.\footnote{Even in this setting, it can be seen that the {\it number} of diseased individuals can still be inferred by Eve. However, this is assumed to be a publicly known/estimable parameter.}

    Given the test outcomes, doctor Bob now desires to identify the set of diseased people. To be able to reconstruct this mapping, a relatively large amount of information (the mapping between individuals' identities and pools tested) needs to be securely communicated from Alice to Bob. As in the one-time pad secure communication setting, this need for a large amount of common randomness makes such schemes unattractive in practice. Nonetheless, the question is theoretically interesting, and some interesting results have been recently reported in this direction by \cite{atallah2008private,goodrich2005indexing,freedman2004efficient,rachlin2008secrecy}.

    An alternative to keeping secret which samples comprise each pool test (encrypting the mixing matrix), is to keep the identity of the inspected patients secret (encrypting the patients’ list). Accordingly, the lab will be able to identify whether the pool-test is infected or not, but will not know the true identity of the patients in each pool-test. For example, Alice can utilize a secret sharing mechanism \cite{shamir1979share} distributing the “secret” (the list of examinees) between all the participating labs such that each lab cannot reconstruct the patients’ list, hence cannot parse the identity of the examinees in each inspected pool-test. The doctor holding the shares from all labs and the results of the pool tests can both reconstruct the patients’ list and identify the infected examinees. Obviously a scheme which relies on patients’ list encryption requires the distribution of the encrypted patient list among all the potential legitimate examiners (doctors).

\item {\it Wiretap secrecy/Physical-layer secrecy:} This is the setting of this paper. Alice does not desire to communicate a large shared key to Bob, and still wishes to maintain secrecy of the identities of the diseased people from ``honest but curious" Eve. Alice therefore does the following two things: (i) For some $\delta \in (0,1)$, she chooses a $1/\delta$ number of independent labs, and divides the T pools to be tested into $1/\delta$ {\it pool sets} of $T\delta$ pools each, and sends each set to a distinct lab. (ii) For each blood pool, she {\it publicly} reveals to {\it all} parties (Bob, Eve, and anyone else who's interested) a set ${\cal S}(t)$ of {\it possible} combinations of individuals whose blood could constitute that specific pool $t$. As to which specific combination from ${\cal S}(t)$ of individuals the pool actually comprises of, only Alice knows {\it a priori} - Alice generates this private randomness by herself, and does not leak it to anyone (perhaps by destroying all trace of it from her records).
    The twin-fold goal is now for Alice to choose pool-sets and set of ${\cal S}(t)$ for each $t \in [T]$ to ensure that as long as no more than one lab leaks information, there is sufficient randomness in the set of ${\cal S}(t)$ so that essentially no information about the diseased individuals identities leaks, but Bob (who has access to the test reports from all the $1/\delta$ labs) can still accurately estimate (using the publicly available information on ${\cal S}(t)$ for each test $t$) the disease status of each individual.
    This scenario closely parallels the scenario in Wyner's Wiretap channel.
    Specifically, this corresponds to Alice communicating a sequence of $T$ test outcomes to Bob, whereas Eve can see only a $\delta$ fraction of test outcomes. To ensure secrecy, Alice injects private randomness (corresponding to which set from ${\cal S}(t)$ corresponds to the combination of individuals that was tested in test $t$) into each test - this is the analogue of the coding schemes often used for Wyner's wiretap channels.
    Note that the setting we suggest herein lets the nurse protect the data in the blood samples, without even having the data at hand, namely, without knowing which item is defective or not, and without sharing any key with the lab or the doctor.
\begin{remark}
It is a natural theoretical question to consider corresponding generalizations of this scenario with other types of broadcast channels from Alice to Bob/Eve (not just degraded erasure channels), especially since such problems are well-understood in a wiretap security context. However, the physical motivation of such generalizations is not as clear as in the scenario outlined above. So, even though in principle the schemes we present in Section II can be generalized to other broadcast channels, to keep the presentation in this paper clean we do not pursue these generalizations here.
\end{remark}
\begin{remark}
Note that there are other mechanisms via which Alice could use her private randomness. For instance, she could {\it deliberately} contaminate some fraction of the tests she sends to each lab with blood from diseased individuals. Doing so might reduce the amount of private randomness required to instantiate secrecy. While this is an intriguing direction for future work, we do not pursue such ideas here.
\end{remark}

\item {\it Public discussion/Public feedback:} The analogue of a public discussion communication scheme in the secure group-testing context is perhaps a setting in which Alice sends blood pools to labs in {\it multiple rounds}, also known as {\it adaptive group testing} in the GT literature. Bob, on observing the set of test outcomes in round $i$, then publicly broadcasts (to Alice, Eve, and any other interested parties) some (possibly randomized) function of his observations thus far. This has several potential advantages. Firstly, adaptive group-testing schemes (e.g.\cite{aldridge2014group}) significantly outperform the best-known non-adaptive group-testing schemes (in terms of smaller number of tests required to identify diseased individuals) in regimes where $K=\omega(N^{1/3})$\footnote{\tifs{As given in \cite{johnson2018performance}, in constant column design (rather than Bernoulli), $K = \omega\left(N^{\log 2 /(1+\log 2)}\right)$}.}. One can hope for similar gains here. Secondly, as in secure communication with public discussion, one can hope that multi-round GT schemes would enable information-theoretic secrecy even in situations where Eve may potentially have access to {\it more} test outcomes than Bob. Finally, such schemes may offer storage/computational complexity advantages over non-adaptive GT schemes. Hence this is an ongoing area of research, but outside the scope of this paper.
\end{itemize}
\subsubsection{Cryptographic secrecy}
As in the context of secure communication, the use of cryptographic primitives to keep information about the items being tested secure has also been explored in sparse recovery problems (again, \cite{atallah2008private,goodrich2005indexing,freedman2004efficient,rachlin2008secrecy}). Schemes based on cryptographic primitives have similar weaknesses in the secure GT context as they do in the communication context, and we do not explore them here.
Moreover, considering the basic problem of blood testing as an application of our model, using physical-layer security (this was, in fact, the first application for non-secure GT), one must remember that, the whole point of grouping tests together is since the testing process itself is expensive. Therefore, the ``nurse" might be able to take the samples, but cannot perform the tests. The tests have to be ``sent" for testing. This leads to two key problems: (i) The materials which are sent for testing are real tubes. They are not bits of information communicated over a channel, therefore standard cryptography cannot be trivially applied to them. (ii) If the tests are distributed to a few labs (e.g., for parallel processing), this naturally leads to a problem of how to make sure each specific lab cannot decode the information from the tests it received. One cannot ``add a certain substance" to the tubes to make sure they cannot be tested or ``encrypt the tubes". Moreover, as we already mentioned in this section, the nurse needs to protect the data in the blood samples, without even having the data at hand (before the samples were actually tested). Hence, schemes based on cryptographic primitives (e.g., AES) will not trivially solve the problem at hand. However, solutions applying cryptographic-security to the labels (names) can be used when the nurse Alice and the doctor Bob are able to share some private data (keys, random permutation, etc.).   %
\ifCH
\newpage
\input{SummaryOfChanges}    %
\fi
\bibliographystyle{IEEE}
\bibliography{references,SecureNetworkCodingGossip}
\end{document}